\documentclass[journal,10pt]{IEEEtran}
\pdfoutput=1 
\usepackage{amsmath,amssymb,amsthm}
\usepackage{graphicx}
\usepackage[usenames]{color}
\usepackage[latin1]{inputenc}
\usepackage{dsfont}
\usepackage{cite}
\graphicspath{{figures/}}
\usepackage{array}
\usepackage{bm}
\usepackage{url}
\usepackage{color}
\usepackage{amsmath,amsthm}
\theoremstyle{remark}
\newtheorem{theorem}{Theorem}

\newtheorem{remark}{Remark}

\usepackage{enumerate}
\usepackage{amssymb}
\usepackage{graphicx}
\usepackage{subfigure}
\usepackage{epsfig}
\usepackage{psfrag}
\usepackage{verbatim}
\usepackage{algorithm}
\usepackage{algorithmic}
\usepackage{threeparttable}
\usepackage{cite}
\usepackage{multicol}
\usepackage{lipsum}

\newcommand{\zhinan}[1]{{\color{black}#1}}
\newcommand{\zhinann}[1]{{\color{black}#1}}
\newcommand{\rmT}{\rm{T}}
\newcommand{\rmD}{\rm{D}}
\newcommand{\rmd}{\rm{d}}

\newcommand{\rmp}{\rm{p}}
\newcommand{\rmH}{\rm{H}}

\newcommand{\tr}{\rm{tr}}

\newcommand{\SNR}{\rm{SNR}}

\newcommand{\bfA}{\mathbf{A}}
\newcommand{\bfb}{\mathbf{b}}
\newcommand{\bfa}{\mathbf{a}}
\newcommand{\Bdiag}{{\rm{Bdiag}}}
\newcommand{\rmvec}{{\rm{vec}}}
\newcommand{\bfI}{\mathbf{I}}
\newcommand{\bfe}{\mathbf{e}}
\newcommand{\bfq}{\mathbf{q}}
\newcommand{\bfh}{\mathbf{h}}
\newcommand{\bfd}{\mathbf{d}}
\newcommand{\bfD}{\mathbf{D}}
\newcommand{\bfy}{\mathbf{y}}
\newcommand{\bfv}{\mathbf{v}}
\newcommand{\bfu}{\mathbf{u}}
\newcommand{\bfU}{\mathbf{U}}
\newcommand{\bfF}{\mathbf{F}}
\newcommand{\bfn}{\mathbf{n}}
\newcommand{\bfr}{\mathbf{r}}
\newcommand{\bfz}{\mathbf{z}}
\newcommand{\bff}{\mathbf{f}}
\newcommand{\bfg}{\mathbf{g}}
\newcommand{\bfG}{\mathbf{G}}
\newcommand{\bfR}{\mathbf{R}}
\newcommand{\bfx}{\mathbf{x}}

\newcommand{\bfw}{\mathbf{w}}
\newcommand{\bfW}{\mathbf{W}}
\newcommand{\mcW}{\mathcal{W}}

\newcommand{\mcT}{\mathcal{T}}
\newcommand{\mcP}{\mathcal{P}}
\newcommand{\mcC}{\mathcal{C}}

\newcommand{\bfLambda}{\mathbf{\Sigma}}
\newcommand{\bflambda}{\boldsymbol{\lambda}}
\newcommand{\bfphi}{\boldsymbol{\phi}}

\newcommand{\bfgamma}{\boldsymbol{\gamma}}
\newcommand{\bfGamma}{\boldsymbol{\eta}}
\newcommand{\Tc}{T_{\rm c}}
\newcommand{\Ts}{T_{\rm s}}

\newcommand{\be}{\begin{equation}}
\newcommand{\ee}{\end{equation}}
\newcommand{\ba}{\begin{array}}
\newcommand{\ea}{\end{array}}
\newcommand{\bdm}{\begin{displaymath}}
\newcommand{\edm}{\end{displaymath}}
\newcommand{\bea}{\begin{eqnarray}}
\newcommand{\eea}{\end{eqnarray}}
\newcommand{\bean}{\begin{eqnarray*}}
\newcommand{\eean}{\end{eqnarray*}}

\def\[{\left[}
\def\]{\right]}
\def\diag{{\rm{diag}}}

\begin{document}
\title{A Time-Variant Channel Prediction and Feedback Framework for Interference Alignment }

\author{
\IEEEauthorblockN{Zhinan Xu, Markus Hofer, Thomas Zemen}
%\IEEEauthorblockA{AIT Austrian Institute of Technology, Vienna, Austria \\
%Email: zhinan.xu@ait.ac.at}

\thanks{Copyright (c) 2015 IEEE. Personal use of this material is permitted. However, permission to use this material for any other purposes must be obtained from the IEEE by sending a request to pubs-permissions@ieee.org.}
\thanks{The authors are with AIT Austrian Institute of Technology, Vienna, Austria (e-mail: \{zhinan.xu, markus.hofer, thomas.zemen\}@ait.ac.at).}%
\thanks{Manuscript received XXX, XX, 2015; revised XXX, XX, 2015.}
}

\markboth{IEEE Transactions on Vehicular Technology, accepted for publication.}{}
\maketitle

\begin{abstract}
In interference channels, channel state information (CSI) can be exploited to reduce the interference signal dimensions and thus achieve the optimal capacity scaling, i.e.~degrees of freedom, promised by the interference alignment technique. However, imperfect CSI, due to  channel estimation error, imperfect CSI feedback and time selectivity of the channel, lead to a performance loss. In this work, we propose a novel limited feedback algorithm for single-input single-output interference alignment in time-variant channels. \zhinan{The feedback algorithm encodes the channel evolution in a small number of subspace coefficients, which allow for reduced-rank channel prediction to compensate for the channel estimation error due to time selectivity of the fading process and feedback delay.} An upper bound for the rate loss caused by feedback quantization and channel prediction is derived. Based on this bound, we develop a dimension switching algorithm for the reduced-rank predictor to find the best tradeoff between quantization- and prediction-error. Besides, we characterize the scaling of the required number of feedback bits in order to decouple the rate loss due to channel quantization from the transmit power. Simulation results show that a rate gain over the traditional non-predictive feedback strategy can be secured and a $60\%$ higher rate is achieved at $20$dB signal-to-noise ratio with moderate mobility.
\end{abstract}

\zhinan{
}

\section{Introduction}
Interference is a crucial limitation in next generation cellular systems. To address this problem, interference alignment (IA) has attracted much attention and has been extensively studied lately. IA is able to achieve the optimal capacity scaling, i.e.~degrees of freedom (DoF), at a high signal-to-noise ratio (SNR) and a rate of $K/2\cdot {\rm{log}}({\rm{SNR}})+o({\rm{log}}({\rm{SNR}}))$ for the $K$ user interference channel with time-variant coefficients \cite{Cadambe2008a}. However, this result is based on the assumption that \emph{global} channel state information (CSI) is perfectly known at all nodes. This is extremely hard to achieve due to the large amount of required feedback information. 

CSI imperfection degrades IA performance. Assuming the imperfect channel matrix is the summation of the true channel matrix and an independent error matrix, the impact of imperfect CSI on IA has been investigated in \cite{Tresch2009,Xie2013,Razavi2014,Aquilina15}. The work of \cite{Tresch2009} derives both upper and lower bounds on the achievable rates assuming noisy CSI. The error performance of IA is studied in \cite{Xie2013} and adaptive schemes are proposed to introduce robustness against CSI imperfection.  The performance loss of IA under CSI mismatch for interference channels is studied in \cite{Razavi2014} showing that full DoF are achievable if the variance of the CSI measurement error is proportional to the inverse of the SNR. Similar results are found in \cite{Aquilina15} for interference broadcast channels.

Limited feedback via quantization is a promising approach to transfer CSI to the transmitter side in frequency division duplex (FDD) systems.
Several approaches address the problem of limited feedback for IA \cite{Bolcskei2009,Krishnamachari2013,Kim2012,Ayach2012,Xu2015aa} assuming perfect channel estimation. In \cite{Bolcskei2009}, channel coefficients are quantized using a Grassmannian codebook for frequency-selective single-input single-output (SISO) channels. The work in \cite{Krishnamachari2013} extends the results to multiple-input multiple-output (MIMO) channels. Both \cite{Bolcskei2009} and \cite{Krishnamachari2013} show that the full DoF are achievable as long as the feedback rate is high enough (which scales with the transmit power). This result aligns with the one that is found for MIMO broadcast channels in \cite{Jindal2006}. The work in \cite{Kim2012} addresses the problem of improving the sum rate under limited feedback by involving additional iterative computation of pre-quantization filters at the receivers. To further reduce the feedback overhead, \cite{Ayach2012} considers differential limited feedback on the Grassmannian manifold by exploiting temporal correlation of the time-selective fading channels. In the context of opportunistic transmission for IA, \cite{Xu2015aa} shows that multi-user diversity can be exploited based on 1-bit feedback from each user, while preserving the full DoF. Instead of quantizing the CSI, \cite{Ayach2012a} considers analog feedback and shows that the DoF of IA can be preserved as long as the forward and reverse link SNRs scale together.

However, for a practical system, the imperfection of CSI is caused by various aspects: 
%It is important to consider the different sources of CSI errors in practical systems:
\begin{enumerate}[(a)]
	\item For time-variant channels, CSI is acquired with the aid of pilot symbols. The channel varies over time due to the mobility of the users. If the channel changes after the transmission of the pilot symbols, the receiver cannot detect the channel variation, which leads to a reduction in sum rate due to the use of outdated channel estimates.
	\item For FDD, CSI is fed back through limited capacity feedback channels. The error due to quantized feedback degrades the IA performance. 
	\item The feedback information arrives at the transmitter with a delay which causes a further performance degradation.
	%\item Overhead, which comes with pilot insertion, does not convey any payload information, leading to a reduction of spectral efficiency.
\end{enumerate}

A related body of research tackling the above mentioned problems exists for single-cell multiuser MIMO systems \cite{Caire2010,Santipach2010}. For interference channels, \cite{ElAyach2012} studies (a) and (b) %and %(d) 
for MIMO IA using a minimum mean square error (MMSE) estimator. The studies in \cite{Caire2010,Santipach2010,ElAyach2012} consider block fading channels. The work of \cite{Mungara2014} extends \cite{ElAyach2012} considering time-selective continuous fading in the payload part, while assuming constant fading for the training part. %Our former publication \cite{Xu2014a} studies (a), (b) and (c) by channel prediction and limited feedback.

In this paper, we jointly consider (a)-(c) for interference alignment in wideband SISO systems with symbol extension over frequency \cite{Cadambe2008a}, i.e.~precoding across orthogonal frequency dimensions. \zhinan{
We address channel estimation, feedback and prediction jointly and analyze the impact of imperfect CSI in terms of sum rate loss. Previous works solely tackle the theoretical aspect of imperfect CSI \cite{Tresch2009,Xie2013,Razavi2014,Aquilina15}, or limited feedback assuming perfect channel estimation \cite{Bolcskei2009,Krishnamachari2013,Kim2012,Ayach2012,Xu2015aa}. Besides, the prediction aspect for IA is studied in time division duplex (TDD) systems \cite{Yu2012} or assuming perfect feedback for FDD systems  \cite{Zhao2014b}. %The work \cite{Tresch2009,Xie2013,Razavi2014,Aquilina15,Bolcskei2009,Krishnamachari2013,Kim2012,Ayach2012,Xu2015aa,Zhao2014b} are based on {\em single-shot systems}.  
However, it is almost impossible to feed back all channel impulse responses of the payload for FDD systems due to limited capacity feedback channels. Our approach takes advantage of the band-limited nature of time-variant channels, which allows us to describe the channel variation of the payload using a few subspace coefficients. The feedback of subspace coefficients using vector quantization for channel prediction is not considered in any of the previous works to the best of our knowledge. We exploit this concept in the context of IA.} Although we present our results for SISO interference channels, the complete strategy can be generalized to MIMO interference channels using a method similar to \cite{Krishnamachari2013} by vectorization of channel matrices. The scientific contributions of the paper: 
\begin{itemize}
\item We tackle the problems (a) and (c) by reduced-rank channel prediction using discrete prolate spheroidal (DPS) sequences \cite{Zemen2007}. Thanks to the energy concentration of the sequences in the Doppler domain, we are able to describe the channel evolution by only a few subspace coefficients. 
\item To address problem (b), we show that the subspace coefficients can be quantized and fed back using vector quantization, which greatly reduces the redundancy of the codebook by exploiting the rotation invariance. In addition, we highlight the importance to feed back the subspace coefficients in delay domain, resulting in a reduction of noise. With the subspace coefficients, the transmitter is able to perform channel prediction to combat the time selectivity of the channel. We generalize the results from \cite{Xu2014a} allowing for more channel extensions than channel delay taps.
\item The subspace vector to be quantized has correlated entries in some cases. We characterize the second order statistics of the subspace vector, which is used for whitening the vector to match the statistics of the quantization codebook, improving the results in \cite{Xu2014a}.
\item An upper bound of the rate loss due to the channel prediction- and quantization-error is derived, which is used to facilitate an adaptive subspace dimension switching algorithm.
\item We show that there exists a tradeoff between quantization error and prediction error at a given feedback rate. The subspace dimension switching algorithm is efficient to capture the tradeoff and find the subspace dimension associated with a higher rate.
\item We characterize the scaling of the required number of feedback bits to decouple the rate loss due to quantization from the transmit power.
\end{itemize}

The following notation is used throughout this paper: We denote a scalar by $a$, a column vector by $\bfa$ and a matrix by $\bfA$. The superscript $^*$, $^{\rmT}$ and $^{\rmH}$ stand for conjugate, transpose and Hermitian transpose, respectively. The notation $\left\| \cdot \right\|$, $\left\| \cdot \right\|_{\rm F}$, $\tr(\cdot)$, $\rmvec(\cdot)$, $\det(\cdot)$, $\lceil{\cdot}\rceil$ and $\mathbb{E}[\cdot]$ denotes vector 2-norm, Frobenius norm, trace, vectorization, determinant, ceiling operation and the expectation operation, respectively. We denote by $[\bfA]_{k,\ell}$ the $(k,\ell)$-th element of matrix $\bfA$ and by $[\bfa]_{k}$ the $k$-th element of vector $\bfa$.
%The $N$-point discrete Fourier transform of a sequence $[a_0,\ldots,a_{N-1}]$ is defined as ${\mathcal{F}_N}\{[a_{0},\ldots,a_{N-1}] \}=[\sum_{n=0}^{N-1} a_n,\sum_{n=0}^{N-1} a_n e^{-i 2\pi \frac{n}{N}},\ldots,\sum_{n=0}^{N-1} a_n e^{-i 2\pi \frac{(N-1)n}{N}}]$.

The rest of the paper is organized as follows: In Section \ref{sec2}, we introduce the system model of IA. Section \ref{RR} presents the reduced-rank channel estimation and prediction algorithm. Section \ref{Trfb} describes the proposed training and feedback algorithm. The performance is quantified in terms of sum rate loss in Section \ref{Ratelossa}, where a subspace switching algorithm is given. The numerical results are provided in Section \ref{numersim}.
We conclude the paper in Section \ref{Conclusion}.

\section{System model} \label{sec2}
Let us consider a $K$ user time- and frequency-selective SISO interference channel, which consists of $K$ transmitter and receiver pairs. We denote by $h_{k,\ell}(t,\tau)$ the time-variant impulse response between transmitter $\ell$ and receiver $k$, where $t$ is time and $\tau$ is delay. Orthogonal frequency division multiplexing (OFDM) is used to convert the time- and frequency- selective channel into $N$ parallel time-selective and frequency-flat channels. The sampled impulse response is defined as $h_{k,\ell}[m,s]=h_{k,\ell}(m\Ts,s\Tc)$, where $1/\Tc$ is the bandwidth and $\Ts=(N+G)\Tc$ denotes the OFDM symbol duration with a cyclic prefix length $G$.
The $S$-tap time-variant sampled impulse response between transmitter $\ell$ and receiver $k$ is denoted by ${\bfh}_{k,\ell}\[m\]=\[ h_{k,\ell}[m,1],\ldots, h_{k,\ell}[m,S] \]^{\rmT}$, $\forall k,\ell \in \{ 1,\ldots,K \}$. Every element of the channel impulse response vector ${\bfh}_{k,\ell}\[m\]$ is independent identically distributed (i.i.d.) with a power delay profile (PDP) $\mathbb{E}\{\bfh_{k,\ell}[m]\bfh_{k,\ell}[m]^{\rmH}\}=\diag(\[p_{k,\ell}^{1},\ldots,p_{k,\ell}^{S}\])$. We assume the channel gain $\sum_{s=1}^{S} p_{k,\ell}^{s}=N$. 

The variation of a wireless channel for the duration of the transmission of a data packet is caused by user mobility and multipath propagation. We define the normalized Doppler frequency of the time-selective fading process $\{h_{k,\ell}[m,s]\}$ as 
\begin{align}
 \nu_{\rm{D}}=f_{\rmD} T_{\rm{s}}
\end{align}
where $f_{\rmD}$ denotes the Doppler frequency in Hertz (Hz). The temporal covariance function over consecutive OFDM symbols becomes
\begin{align}
R_{\bfh_{k,\ell}}[m]=\mathbb{E}\{\bfh_{k,\ell}[a]^{\rmH}\bfh_{k,\ell}[a+m]\}.
\end{align}
The temporal covariance matrix is defined as $[\bfR_{\bfh_{k,\ell}}]_{a,m}=R_{\bfh_{k,\ell}}[a-m]$ for $a,m \in [0,\ldots,M-1]$.

 %First each $S$ dimensional vector $\bfh_{k,\ell}[m]$ is converted to a $N$ dimensional vector through zero padding.
%It then applies $N$-point discrete Fourier transform such that the ${N \times N}$ matrix of 
Using OFDM, The observed frequency selective channel can be converted into $N$ narrowband frequency-flat channels as 
\begin{align}
{\bfw}_{k,\ell}[m]= \bfD_{N \times S} {\bfh}_{k,\ell}[m] \label{channelconvert}
\end{align}
%\begin{align}
%	{\bfw}_{k,\ell}[m]= {\mathcal{F}_N}\{[{\bfh}_{k,\ell}[m]^{\rmT}, \mathbf{0}_{1 \times (N-S)}]^{\rmT} \} . 
%\end{align} 
where $\bfD_{N \times S}$ is the $N \times S$ submatrix of the $N\times N$ DFT matrix $\bfD_{N}$. The DFT matrix $\bfD_{N}$ is defined as
$[\bfD_{N}]_{i,j}=\frac{1}{\sqrt N}e^{-j2\pi (i-1)(j-1)/N}$, $\forall i,j \in \{ 1,\ldots,N \}$. The diagonal matrix containing the channel frequency response can be written as $\bfW_{k,\ell}[m] =\diag \left( \bfw_{k,\ell}[m] \right)$.

We consider two different communication phases: (i) the CSI acquisition via pilots and (ii) the transmission of payload. In the CSI acquisition phase, the pilot symbols from different transmitters are orthogonalized in time. During the transmission of payload, all transmitters will send simultaneously. However, for a given transmitter, its signal is only intended to be received by a single user for a given signaling interval. The signal received is the superposition of the signals transmitted by all transmitters. The received signal at receiver $k$ in these two phases can thus be modeled by
\begin{align}
\bfy_{k}[m]=
\begin{cases}
\bfW_{k,k}[m] \bfx_{k}[m]+\bfn_{k}[m]\,, & m \in \mcP_k\\
\bfW_{k,k}[m] \bfx_{k}[m]+\nonumber \\
\quad \quad \sum_{k\neq \ell} {\bfW_{k,\ell}\[m\]\bfx_{\ell}[m]}+ \bfn_{k}[m] \,, & \text{elsewhere}
\end{cases}
\end{align}
where $\mathcal P_k$ denotes the pilot position indices of user $k$. The vector $\bfx_{k}[m] \in \mathbb{C}^{N \times 1}$ denotes the transmitted symbol for user $k$ with power constraint $\mathbb{E}{\{\bfx_{k}[m]^{\rmH}\bfx_{k}[m]\}}=PN$, where $P$ is the transmit power per subcarrier. Additive complex symmetric Gaussian noise at receiver $k$ is denoted by $\bfn_{k}[m] \sim  \mathcal{C} \mathcal{N}(0, \mathbf{I}_{N})$. The SNR is defined as ${\SNR}=P$.

In this work we consider a user velocity and carrier frequency such that the Doppler bandwidth of the fading process $f_D$ is much smaller than the subcarrier spacing $ f_{\rm sc}=\frac{1}{\Tc N}$. \zhinan{Hence, the inter-carrier interference resulting from Doppler shift is small enough to be neglected for the processing at the receiver side, see the discussion in \cite[Sec.~II]{edfors1996} and \cite[Sec.~II]{Zemen2012}.

Throughout the paper, we adopt a widely used assumption, where all channels ${\bfh}_{k,\ell}~ \forall k, \ell$ have the same gain and time-selective fading statistics. Transmit power $P$ is assumed for all transmitters. 
}
 
\subsection{SISO Interference Alignment with Perfect CSI}
IA can achieve optimal DoF when infinite channel extensions exist\cite{Cadambe2008a}. Using IA over $N$ orthogonal subcarriers, each transmitter $k$ sends a linear combination of $d_k<N$ symbols $s_k^{i}[m]$, along the linear precoding vectors $\bfv^i_k \in \mathbb{C}^{N \times 1}$, yielding 
\be
\bfx_{k}[m]=\sum_{i=1}^{d_k}{\bfv_k^{i}[m]s_k^{i}[m]}\,,
\ee
where $s_k^{i}[m] \in \mathbb{C}$ denotes the transmitted symbol and $\mathbb{E}\{\left|s_k^{i}[m]\right|^2\}= PN/d_k$. The precoding vector $\bfv_k^{i}[m]$ fulfills $\left\|\bfv_k^{i}[m] \right\|^2=1$. 
Defining the decoding vector $\bfu_k^{i}[m] \in \mathbb{C}^{N \times 1}$ subject to $\left\| \bfu_k^{i}[m]\right\|^2=1$, the received signal at receiver $k$ for symbol $i$ can be expressed as 
\begin{align}
&\bfu_k^{i}[m]^{\rmH} \bfy_{k}[m] = \underbrace{\bfu_k^{i}[m]^{\rmH}  \bfW_{k,k}[m] {\bfv_k^{i}[m]s_k^{i}[m]}}_{\text{desired signal}} +\nonumber\\
&\underbrace{\bfu_k^{i}[m]^{\rmH} \sum_{j \neq i} \bfW_{k,k}[m] {\bfv_k^{j}[m]s_k^{j}[m]}}_{\text{inter-stream interference} }+\nonumber\\
& \underbrace{\bfu_k^{i}[m]^{\rmH}\sum_{\ell \neq k}\sum_{j=1}^{d_\ell} {\bfW_{k,\ell}\[m\]{\bfv_\ell^{j}[m]s_\ell^{j}[m]}}}_{\text{inter-user interference} }+\bfu_k^{i}[m]^{\rmH} \bfn_{k}[m]
\end{align}
for $i \in \{1,\ldots,d_k\}$ and $k \in \{1,\ldots,K\}$. Considering i.i.d Gaussian input of $s_k^{i}[m]$, the achievable sum rate is given by
\begin{align}
& R_{\rm sum}[m] \nonumber\\
&=\sum_{k,i} \frac{1}{N} \log_{2}\left( 1+ \frac{\displaystyle \frac{NP}{d_k}\left| \bfu_k^{i}[m]^{\rmH} \bfW_{k,k}[m] \bfv_k^{i}[m]\right|^2  }{\mathcal{I}_{k,i}^{1}[m]+\mathcal{I}_{k,i}^{2}[m]+ 1}\right).
\end{align}
where
\begin{align}
{\mathcal{I}}_{k,i}^{1}[m]&=\sum_{j\neq i} \frac{NP}{d_k} \left|   {\bfu}_k^{i}[m]^{\rmH} \bfW_{k,k}[m] {\bfv}_k^{j}[m] \right|^2, {\rm{~and}}\label{Inter1}\\ 
{\mathcal{I}}_{k,i}^2 [m]&=\sum_{\ell \neq k}\sum_{j=1}^{d_\ell} \frac{NP}{d_\ell} \left|   {\bfu}_k^{i}[m]^{\rmH} \bfW_{k,\ell}[m] {\bfv}_\ell^{j}[m]\right|^2, \label{Inter2}
\end{align}
denote inter-stream interference and inter-user interference, respectively.

 The precoding and decoding vectors can be designed according to \cite{Cadambe2008a}. Each transmitter computes the precoding vectors $\bfv_k^{i}[m]$ such that the interference signals from the undesired $K-1$ transmitters are aligned at all receivers leaving the interference free subspace for the intended signal.  With perfect CSI, the following IA conditions should be satisfied
\begin{align}
\label{cond1}
\bfu_k^{i}[m]^{\rmH} \bfW_{k,k}[m] \bfv_k^{j}[m]&=0, \quad \forall k,~\forall i\neq j \\ 
\label{cond2}
\bfu_k^{i}[m]^{\rmH} \bfW_{k,\ell}[m] \bfv_\ell^{j}[m]&=0, \quad \forall k\neq \ell,~\forall i,~j \\ 
\label{cond3}
\left|\bfu_k^{i}[m]^{\rmH} \bfW_{k,k}[m] \bfv_k^{i}[m]\right| &\geq c>0,  \quad \forall k,~i
\end{align}
where $c$ is a constant. Accordingly, the interference terms can be perfectly canceled satisfying $\mathcal{I}_{k,i}^{1}[m]=\mathcal{I}_{k,i}^{2}[m]=0$.

\section{Reduced-Rank Channel Estimation and Prediction} \label{RR}
In this section, we introduce the idea of channel prediction. First, a well-known minimum mean square error (MMSE) solution is given. In subsection \ref{RRCP}, we present the reduced-rank predictor and its relation to the MMSE solution. To simplify notations, we drop the indices of transmitters and receivers and focus on the prediction problem for a specific subcarrier. Let us denote by $w[m,n]$, $n[m,n]$ $y[m,n]$ and $x[m,n]$ the $n$-th element of the vector $\bfw[m]$ $\bfn[m]$ $\bfy[m]$ and $\bfx[m]$, respectively. 
The channel samples of the $n$-th subcarrier over time can be written as 
\begin{align}
\bfg^{n}=\left[w[0,n],\ldots,w[M-1,n]\right]^{\rmT}, \label{eqtvch}
%\bfg^{n}=\bigg[\big[\bfw \left[0 \right]\big]_n,\ldots,\big[\bfw\left[M-1\right]\big]_n\bigg]^{\rmT}, 
\end{align}
where $M$ is the length of a single data block. 

A number of $N_{\rm P}$ pilot symbols $x[m,n]\in \{\sqrt{P},-\sqrt{P}\}, \forall m \in \mcP$ known at the receivers allow us to acquire channel knowledge. With the pilot symbols, we obtain the noisy channel observations at $m \in \mathcal P$ according to
$w'[m,n]=\frac{1}{P} y[m,n]x[m,n]=w[m,n]+\frac{1}{\sqrt{P}}n'[m,n]$,
%$\big[\bfw'[m]\big]_n=\big[\bfw[m]\big]_n+\frac{1}{\sqrt{P}}\big[\bfn'[m]\big]_n$
where
%$n'^{n}[m]=n^{n}[m]x^{n}[m]^{*}$
$n'[m,n]=\frac{1}{\sqrt{P}}{n[m,n]x[m,n]}$ has the same statistical properties as %$n^{n}[m]$
$n[m,n]$.
The noisy observation vector of the $n$-th subcarrier over time 
\begin{align}
\bfg'^{n}=\bigg[w' \left[0,n \right],\ldots,w'\left[M-1,n\right]\bigg]^{\rmT}
\end{align}
is used for channel prediction. Defining the  $M \times 1$ vector $\bfr_{\bfh}[m]  = [R_{\bfh}[m],  R_{\bfh}[m-1],\ldots, R_{\bfh}[m-M+1] ]$, the estimator minimizing the MSE can be derived as \cite{Kay1993}
\begin{align}
\tilde{w}_{\rm MMSE} [m,n]= \bfr_{\bfh}^{(\mcP)}[m]^{\rmH} (\bfR_{\bfh}^{(\mcP)}+\frac{1}{P}\bfI_{N_{\rm P}})^{-1}\bfg'^{{n}{(\mcP)}}
\end{align}
where 
the covariance matrix $\bfR_{\bfh}^{(\mcP)} \in \mathbb{C}^{N_{\rm P} \times N_{\rm P}}$ of the channel at pilot positions is obtained as a sub-matrix of $\bfR_{\bfh} \in $ by extracting $K$-spaced rows and/or columns, i.e.~$[\bfR_{\bfh_{k,\ell}}^{(\mcP)}]_{i,m}=[\bfR_{\bfh_{k,\ell}}]_{K(i-1)+i,K(m-1)+m}$.
The vectors $\bfg'^{n (\mcP)}$ contains the respective elements for $m \in {\mcP}$ in the same order as in (\ref{eqtvch}). The $N_{\rm P} \times 1$ vector $\bfr_{\bfh}^{(\mcP)}[m]$ contains the respective elements of $R_{\bfh}[m-m_{\rm P}]$ for $m_{\rm P} \in \mcP$ in the same order as $\bfr_{\bfh}[m]$.

\subsection{Reduced-Rank Channel Predictor} \label{RRCP}
The channel $\bfg^{n}$ can be approximated by a reduced rank representation \cite{Zemen2007,Dietrich2005}, which expands $\bfg^{n}$ by $D$ orthonormal basis functions $\bfu_{p}=[u_{p}[0],\ldots,u_{p}[M-1]]^{\rmT}$, $p \in \{0,\ldots,D-1\}$
\be
\bfg^{n} \approx \bfU \bfphi^{n} =\sum_{p=0}^{D-1} \phi_p^n \bfu_p\,,
\ee
where $\bfU=[\bfu_0,\ldots,\bfu_{D-1}]$ collects $D$ basis vectors of a temporal covariance matrix $\bfR_{\bfh}$ and $\bfphi^n = [\phi_0^n,\ldots,\phi_{D-1}^n]$ contains the subspace coefficients for the channel $\bfg^n$.

 Let us define $\bff[m]=[u_{0}[m],\ldots,u_{D-1}[m]]^{\rmT}$, which collects the values of the basis functions at time $m$. The estimate of ${\bfphi}^n$ can be calculated according to
%\be
%\tilde{\bfphi}^n= \bfU^{\rmH} {\bfg'}^n. \label{cal_phi} 
%\ee
\begin{align}
\tilde{\bfphi}^n&= \bfG^{-1} \sum_{m \in \mcP} w'[m,n] \bff[m]^{\ast}\,, \label{cal_phi}  \\
&=\bfG^{-1} {\bfU^{(\mcP)}}^{\rmH} \bfg'^{n (\mcP)} \label{obtianophi}
\end{align}
where $\bfG=\sum_{m \in \mcP} \bff[m]\bff[m]^{\rmH}={\bfU^{(\mcP)}}^{\rmH}{\bfU^{(\mcP)}}$ and $\bfU^{(\mcP)}=[\bfu_{0}^{(\mcP)},\ldots,\bfu^{(\mcP)}_{D-1}]$. The vector $\bfu^{(\mcP)}_p$ contains the respective elements for $m \in {\mcP}$ in the same order as in (\ref{eqtvch}). Thus, the estimated (predicted) $n$-th subchannel at time instant $m \in \mathbb{Z}$ is given by
\begin{align}
\tilde{w} [m,n]= \sum_{p=0}^{D-1}{\tilde{\phi}_p^n u_p[m]}=\bff[m]^{\rmT} \tilde{\bfphi}^n.  \label{pred}
\end{align}

\subsection{The Choice of Subspace Dimension - An Upper Bound}
In wireless communication systems, detailed second-order statistics are difficult to obtain due to the short time-interval over which the channel can be assumed to be stationary \cite{Bernad2013}. For this reason, we assume incomplete second-order statistics in this work, where only the support $\mcW=(-\nu_{\rm{D}},\nu_{\rm{D}})$ of the Doppler spectrum is known to the transmitters and receivers with $\nu_{\rm{D}} \ll 1/2$. For the case of unknown support, please refer to adaptive channel estimation using hypotheses test \cite{Zemen2012,Hofer15}. The shape of the Doppler spectrum is assumed to be flat with support $\mcW$, which is given by 
\begin{align}
S_{\rm flat}(\nu,\mcW)=
\begin{cases}
\frac{1}{| \mcW |},&  \nu \in \mcW\\
0,&  {\rm otherwise}.
\end{cases}
\end{align}
The covariance function of such a fading process becomes 
\begin{align}
R_{\rm flat}[m,\mcW]=\frac{\sin(2\pi m \nu_{\rm{D}})}{\pi m |\mcW|}.
\end{align}
The corresponding covariance matrix is $\left[\bfR_{\rm flat}[\mcW]\right]_{l,m}=R_{\rm flat}[l-m,\mcW]$ for $l,m \in [0,\ldots,M-1]$. 
The eigenvector of the covariance matrix $R_{\rm flat}[m,\mcW]$ are also known as DPS sequences \cite{Thompson1982,Slepian1978}, which are utilized as the basis functions $\bfu_{p}[\mcW]$ in this paper. %the basis functions $\bfu_{p}[\mcW]$ are chosen to be the eigenvectors of $\bfR_{\rm flat}[\mcW]$, which are also known as DPS sequences \cite{Thompson1982,Slepian1978}. 
The band-limiting region of the DPS sequences $\bfu_{p}[\mcW]$ is chosen according to the support $\mcW$ of the Doppler spectrum of the time-selective fading process. To ease the notation, we drop $\mcW$ in the rest of the paper. Given $\bfu_{p}$,  \cite [Sec.~III.D]{Zemen2007} shows that the DPS sequences can be extended over $\mathbb{Z}$ in the minimum-energy band-limited sense, enabling channel prediction in (\ref{pred}). 
The energy of the DPS sequences is most concentrated in the interval of block length $M$. This energy concentration is defined as
\be
\kappa_{p}=\frac{\displaystyle \sum_{m=0} ^{M-1} \left| u_{p}[m]\right|^2} {\displaystyle \sum_{m=-\infty} ^{\infty} \left| u_{p}[m]\right|^2}.
\ee
The values  $\kappa_{p}$ are clustered near 1 for $p\leq \left\lceil 2 \nu_{\rm{D}} M \right\rceil$ and decay rapidly for $p > \left\lceil 2 \nu_{\rm{D}} M \right\rceil$. The optimal subspace dimension that minimizes the mean square error (MSE) for a given noise level is found to be
\cite{Zemen2007}
\be
\mathcal{D}_{\rm{ub}}=  \underset{D \in \{1,\ldots,M \}}{\arg \min} \left( \frac{1}{\left| \mcW \right| M}  \sum_{p=D}^{M-1} \kappa_{p} + \frac{D}{MP} \right). \label{optD}
\ee
Later in Sec.~\ref{Ratelossa} we will see that $\mathcal{D}_{\rm{ub}}$ is the upper bound of the subspace dimension when quantized feedback is used.

% use the channel prediction method presented in \cite{Zemen2007},  which employs index-limited DPS sequences  \cite{Slepian1978} to form the orthogonal basis vectors $\bfu_{p}$. 

%The band-limiting region of the DPS sequences $\bfu_{p}[\mcW]$ is chosen according to the support $\mcW$ of the Doppler spectrum of the time-selective fading process, where $\mcW=(-\nu_{\rm{D}},\nu_{\rm{D}})$ with $\nu_{\rm{D}} < 1/2$.

\begin{remark} \label{predrelation}
The reduced-rank channel prediction is a close approximate of the MMSE predictor, especially at high SNR. At high SNR ($P\rightarrow\infty$), the MMSE predictor converges to a maximum-likelihood (ML) predictor i.e., $\tilde{w}_{\rm ML} [m,n]= \bfr_{\bfh}^{(\mcP)}[m]^{\rmH} \bfR_{\bfh}^{{-1}{(\mcP)}}\bfg'^{{n}{(\mcP)}}$. For the reduced-rank predictor, more basis functions tend to be taken as $P\rightarrow\infty$ according to (\ref{optD}). Therefore, it also converges to a ML predictor due to the relationship shown in \cite[Eq.38]{Zemen2007}.
\end{remark}

\section{Training and Feedback for IA} \label{Trfb}
In this section, we consider a limited feedback scheme for the subspace coefficients $\tilde{\bfphi}^n$ estimated at the receiver side. Figure~\ref{frames} shows the working principle of the feedback system. The subspace coefficients are estimated using the pilot symbols. Each receiver estimates the channels to all $K$ transmitters separately. To this end, the pilot symbols from different transmitters are orthogonalized in time. The number of pilot symbols  $M_{\rmp}$ for each transmitter satisfies $M=KM_{\rmp}$. The pilot placement for the $k$-th transmitter is defined as 
\begin{align}
{\mathcal P}_k = \big\{k+(i-1)K {\text{, where }  } i\in \{1,\ldots, M_{\rmp}\} \big\}.
\end{align}

Error-free dedicated broadcast channels with delay $T_{\rm D}$ are assumed from each receiver to all the other nodes, i.e.~ all the transmitters and all other receivers. During the feedback phase, each receiver broadcasts the estimated subspace coefficients using $N_{\rm d}$ bits. Upon reception of the quantized feedback, the transmitters and receivers can calculate the IA precoders and decoders, respectively.
\begin{figure}[t]
	\centering
	\includegraphics[width=0.8\columnwidth]	{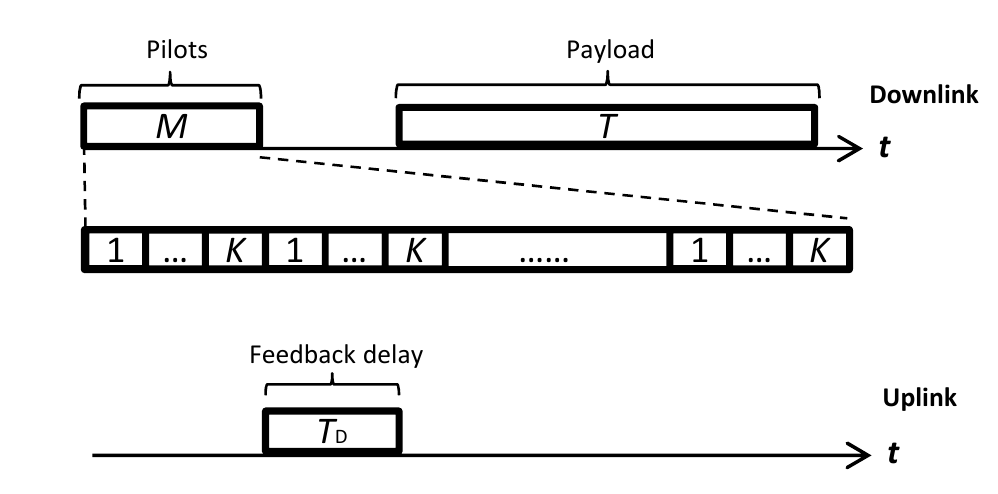}
	\caption{Signaling model, where $M$ denotes the length of the pilot sequence and $T$ the payload length.} 
	\label{frames}
\end{figure}

\subsection{Noise Reduction}\label{noiosereduction}
Assuming a wide-sense-stationary fading process, the $N$ narrowband channels  from the same transmitter receiver pair have the same Doppler bandwidth, thus all $N$ fading processes share the same set of basis functions. 
Due to the fact that $N \geq S$, the impulse response $\bfh[m]$ contains less coefficients than the frequency response $\bfw[m]$. Thus, $\bfh[m]$ is better suited for CSI feedback. 
\zhinan{
For a noisy channel, the equivalent subspace coefficients in the delay domain can be expressed as
\begin{align}
[\tilde{\bfgamma}^1 ,\ldots, \tilde{\bfgamma}^N]^{\rmT}= \bfD_{N}^{\rmH} [\tilde{\bfphi}^1 ,\ldots, \tilde{\bfphi}^N ]^{\rmT}  \label{linear}.
\end{align}
Among all $N$ taps of the estimated channel, only the first $S$ columns $\{\tilde{\bfgamma}^1 ,\ldots, \tilde{\bfgamma}^S\}$ are from the contribution of the channel, while the rest $N-S$ columns are due to noise impairment. Hence, the elimination of these channel taps can improve the SNR. This can be seen by rewriting (\ref{cal_phi}) as  
\begin{align}
\tilde{\bfphi}^n&={\bfphi}^n+\ddot{\bfphi}^n\\
&= \bfG^{-1} \sum_{m \in \mcP} w[m,n] \bff[m]^{\ast}+\bfG^{-1} \sum_{m \in \mcP} n'[m,n] \bff[m]^{\ast}
%&= \bfG^{-1} \sum_{t \in \mcP} w^{n}[m] \bff[m]^{\ast}+\bfG^{-1} \sum_{t \in \mcP} n'^{n}[m] \bff[m]^{\ast}\,\\
\end{align}
where the subspace coefficients $\tilde{\bfphi}^n$ are decomposed into two parts, corresponding to the channel and noise, respectively. The delay domain coefficients can also be decomposed into two parts as
\begin{align}
[\tilde{\bfgamma}^1 ,\ldots, \tilde{\bfgamma}^N]=[{\bfgamma}^1 ,\ldots, {\bfgamma}^S,\mathbf{0}_{D \times (N-S)} ]+[\ddot{\bfgamma}^1 ,\ldots, \ddot{\bfgamma}^N] \label{decouplenoisecontr}
\end{align}
where
%For the first part, the equivalent delay domain representation can be obtained as
 \begin{align}
[{\bfgamma}^1 ,\ldots, {\bfgamma}^S,\mathbf{0}_{D \times (N-S)} ]^{\rmT}= \bfD_{N}^{\rmH} [{\bfphi}^1 ,\ldots, {\bfphi}^N ]^{\rmT},
 \end{align}
 and
 \begin{align}
 [\ddot{\bfgamma}^1 ,\ldots, \ddot{\bfgamma}^N]^{\rmT}=\bfD_{N}^{\rmH} [\ddot{\bfphi}^1 ,\ldots, \ddot{\bfphi}^N ]^{\rmT}.  \label{linear1}
 \end{align}
Each vector ${\bfgamma}^s$ contains the subspace coefficients corresponding to the $s$-th channel tap ${h}^{s}[m]$. The entries of $\ddot{\bfphi}^n$ are a set of random Gaussian variables with the covariance matrix
\begin{align}
\mathbb{E}\{\ddot{\bfphi}^n{\ddot{\bfphi}^{n \rmH}}\}&=\bfG^{-1} {\bfU^{(\mcP)}}^{\rmH}  \left( \frac{\bfI_{\frac{M}{K}}}{P} \right) {\bfU^{(\mcP)}} \bfG^{-1} \\
&=\frac{1}{P}\bfG^{-1} \label{noisecoeffmean}.
\end{align}
%where $\bfG$ is a diagonal matrix. 
Due to uncorrelated noise over $N$ subchannels, $\{\ddot{\phi}_d^1,\ldots,\ddot{\phi}_d^N\}$ is a set of uncorrelated random Gaussian variables with the same distribution. The same is also true for $\{\ddot{\gamma}_d^1,\ldots,\ddot{\gamma}_d^N\}$ according to Wiener-Khinchine theorem. 
Since only $\{\tilde{\bfgamma}^1 ,\ldots, \tilde{\bfgamma}^S\}$ are relevant to the actual channel realizations and thus have to be fed back to the transmitters.
%\begin{figure}[t]
%	\centering
%	\includegraphics[width=1\columnwidth]	{ch_est_flow_diagram}
%	\caption{Flow diagram for channel estimation and channel tap elimination.} 
%	\label{chestflowdia}
%\end{figure}
By neglecting $\{\tilde{\bfgamma}^{S+1} ,\ldots, \tilde{\bfgamma}^N\}$, the SNR can be increase by a factor of $N/S$.

To establish a tractable analysis, we assume that the number of delay taps $S$ is known at the receiver side. For practical wireless channels, the number of taps can be estimated via most significant tap detection \cite{Minn2000,Wan2010}.
}

\subsection{Reformulation of Subspace Representation for the SISO Interference Channels}
With the basis coefficients $\{\tilde{\bfgamma}^1 ,\ldots, \tilde{\bfgamma}^S\}$ obtained from (\ref{linear}), the predicted channel impulse response can be calculated as
\be
\begin{aligned}
	\tilde{\bfh}[m] & =[\tilde{\bfgamma}^1 ,\ldots, \tilde{\bfgamma}^S] ^{\rmT} \bff[m] \\
	&=\bfF[m]  {\tilde{\bfGamma}} \label{recstCh}
\end{aligned}
\ee
where $\bfF \in \mathbb{C}^{S \times DS}$ and $\tilde{\bfGamma} \in \mathbb{C}^{DS \times 1}$ are defined as follows:
\zhinan{
\begin{align}
	\bfF[m]%&=\Bdiag \left( \bff[m]^{\rmT}, \ldots,\bff[m]^{\rmT} \right)\\
	&=\left[\begin{array}{cccc}
		\bff[t]^{\rmT} & \mathbf{0}_{D \times 1}^{\rmT} & \ldots & \mathbf{0}_{D \times 1}^{\rmT} \\
		\mathbf{0}_{D \times 1}^{\rmT}  & \bff[t]^{\rmT}  & \ldots & \mathbf{0}_{D \times 1}^{\rmT} \\
		\vdots & \vdots & \ddots  & \vdots \\
		\mathbf{0}_{D \times 1}^{\rmT} &  \mathbf{0}_{D \times 1}^{\rmT} & \ldots & \bff[t]^{\rmT}
		\end{array}\right]
\end{align}
}
and
\begin{align}
	\tilde{\bfGamma}=   {\left[\begin{array}{ccc}
			\tilde{\bfgamma}^1 \\
			\vdots\\
			\tilde{\bfgamma}^S \\
		\end{array}\right]}.
\end{align}
Note that the vector $\tilde{\bfGamma}$ is not unique since the achievable rate for IA is invariant to a norm change and phase rotation of $\tilde{\bfGamma}$ (this will be shown in Sec.~\ref{Ratelossa}). Therefore, it is equivalent to know $\tilde{\bfGamma}$ or $\alpha \tilde{\bfGamma}$ at the transmitter side, where $\alpha \in \mathbb{C}$. Thus, the CSI feedback problem becomes feeding back a point on the Grassmannian manifold $\mathcal{G}_{DS,1}$, which can be realized by vector quantization.

%The diagonal matrix $\bfLambda=\Bdiag({\bflambda,\ldots,\bflambda})$ contains the mean of the subspace coefficients, where 
%\begin{align}
%	\bflambda=\diag \left( \bfG^{-1} {\bfU^{(\mcP)}}^{\rmH}  \left( \bfR_{\bfh}^{(\mcP)}+\frac{1}{P}\bfI_{\frac{M}{K}}  \right) {\bfU^{(\mcP)}} \bfG \right).
%\end{align}
%The covariance matrix $\bfR_{\bfh}^{(\mcP)}$ of the channel at pilot positions is obtained as a sub-matrix of $\bfR_{\bfh}$ by extracting $K$-spaced rows and/or columns, i.e.~$[\bfR_{\bfh_{k,\ell}}^{(\mcP)}]_{i,m}=[\bfR_{\bfh_{k,\ell}}]_{K(i-1)+i,K(m-1)+m}$.

\subsection{Vector Quantization and Feedback}
After the subspace coefficient vector $\tilde{\bfGamma} \in \mathbb{C}^{DS \times 1}$ is obtained at the receiver side, the receiver quantizes the vector according to its codebook and broadcasts the index to the transmitter side through a feedback channel using $N_{\rm d}$ bits. 

If the vector $\tilde{\bfGamma}$ has correlated entries, the design of the optimal codebook is difficult. In the special case with i.i.d entries, the optimal codebook for quantization can be generated numerically using the Grassmannian line-packing approach \cite{Love2003,Mukkavilli2003}. \zhinan{However, it is still challenging to find the optimal codewords which achieve the quantization bound promised by \cite{Dai2008}, except for some specific cases (e.g.~lower dimensional packing problems, such as vector quantization problems, with small codebook size \cite{Pitaval2011}).} To overcome this, random vector quantization (RVQ) codebooks are proposed. The codewords of $\mcC_{\rm rnd}$ are independent unit-norm vectors from the isotropic distribution on the complex unit sphere \cite{Au-yeung2007,Jindal2006}. RVQ is commonly used to analyze the effects of quantization because it is mathematically tractable and asymptotically optimal with a distortion on the order of $2^{-\frac{N_{\rm d}}{DS-1}}$. %Due to the isotropic nature of RVQ codebooks,  

When the vector to be quantized has correlated entires, a better codebook can be designed by skewing the RVQ codebook to match the correlation structure \cite{Giannakis2006,Love2006}. The skewed codebook yields performance gain compared to the RVQ codebook. However, the exact characterization of the quantization error has remained an open question. A recent study in \cite{Raghavan13} tries to derive the SNR loss for single user MIMO beamforming system using skewed codebooks, but it still remains in terms of expectations of eigenvalues.

% according to $\bfR_{{\tilde\bfGamma}}={\mathbb E}{[\tilde{\bfGamma}\tilde{\bfGamma}^{\rmH}]}$, where $\bfR_{{\tilde\bfGamma}}$ is not an identity matrix
%To overcome the difficulty in obtaining analytical results, we consider two types of quantizers: one utilizes RVQ codebook, which will facilitate our analytical performance analysis in Sec.~\ref{Ratelossa}; the other leverages the benefit of skewed codebooks.

To overcome the difficulty in obtaining analytical results, we utilize a RVQ codebook, which will facilitate our analytical performance analysis in Sec.~\ref{Ratelossa}. 
%\subsubsection{Quantization using RVQ codebooks}
Considering a factorization of the covariance matrix $\bfR_{{\tilde\bfGamma}}=\bfLambda \bfLambda^{\rmH}$, the vector ${\tilde\bfGamma}$ can be whiten as
\begin{align}
\breve\bfGamma&=\bfLambda^{-1} {\tilde\bfGamma}.
% %&=\bfA \bfSigma^{-\frac{1}{2}} {\tilde\bfGamma}
\end{align}
The covariance matrix $\bfR_{{\tilde\bfGamma}}=\Bdiag({\bflambda^1,\ldots,\bflambda^S})$ contains the covariance matrix of the subspace coefficients for each tap $s$, i.e.
\begin{align}
	\bflambda^s&=\mathbb{E}\{\tilde{\bfgamma}^s{\tilde{\bfgamma}^{s \rmH}}\}\\
	&= \mathbb{E}\{ [\tilde{\bfphi}^1 ,\ldots, \tilde{\bfphi}^N ] \bfd_{s}^{*} \bfd_{s}^{\rmT}  [\tilde{\bfphi}^1 ,\ldots, \tilde{\bfphi}^N ]^{\rmH}\} \label{phid}\\
	&= \bfG^{-1} {\bfU^{(\mcP)}}^{\rmH}\mathbb{E} \{[\bfg'^{1} ,\ldots, \bfg'^{N}  ] \bfd_{s}^{*} \nonumber\\ &\quad \quad  \bfd_{s}^{\rmT}  [\bfg'^{1} ,\ldots, \bfg'^{N} ]^{\rmH}\} {\bfU^{(\mcP)}} \bfG^{-1} \label{phid1}\\	
	&= \bfG^{-1} {\bfU^{(\mcP)}}^{\rmH}  \left( p^{s}\bfR_{\bfh}^{(\mcP)}+\frac{1}{P}\bfI_{\frac{M}{K}}  \right) {\bfU^{(\mcP)}} \bfG^{-1}. \label{mmsepre}
\end{align}
where $\bfd_{s}$ is the $s$-th column of the DFT matrix $\bfD_{N \times S}$. The equations (\ref{phid}), (\ref{phid1}) and (\ref{mmsepre}) are due to (\ref{linear}), (\ref{obtianophi}) and (\ref{channelconvert}), respectively.

We do not assume exact knowledge of the PDP in our work, instead a flat PDP assumption
\begin{align} {\mathbb{E}\{\bfh[m]\bfh[m]^{\rmH}\}=\diag(\[p^{1},\ldots,p^{S}\])}=\frac{N}{S}{\bfI_S}
\end{align}
is used. In case the assumed PDP matches the PDP of the true channel, the vector $\breve\bfGamma$ will be isotropically distributed with uncorrelated entries. 

%Due to the lack of PDP knowledge, a mismatched

 The corresponding RVQ codebook  $\mcC_{\rm rnd}$ contains $2^{N_{\rm d}}$ unit-norm vectors, i.e.~$\mcC_{\rm rnd}=\{\hat{\bfGamma}_1,\ldots,\hat{\bfGamma}_{2^{N_{\rm d}}} \}$. Using codebook $\mcC_{\rm rnd}$, the quantized version of $\tilde \bfGamma$ can be obtained as
\begin{align}
\hat \bfGamma= \underset{\hat{\bfGamma}_i \in \mcC_{\rm rnd}}{\arg \max} |\hat{\bfGamma}_i^{\rmH} \breve\bfGamma |.
\end{align}
After receiving the feedback information, the transmitters can reconstruct the quantized vector by adding the correlation $\bfLambda \hat \bfGamma$.

%\subsubsection{Quantization using skewed codebooks}
%We consider a fixed skewing matrix, which biases the isotropic vectors in the RVQ codebook and directs them towards its singular vectors, yielding the skewed codebook
%\bea
%\mcC_{\rm sk}=\left\{\frac{\bfLambda\hat{\bfGamma}_1}{\|\bfLambda\hat{\bfGamma}_1\|},\ldots,\frac{\bfLambda\hat{\bfGamma}_{2^{N_{\rm d}}}}{\bfLambda\hat{\bfGamma}_{2^{N_{\rm d}}}} \right\}.
%\eea
%Each vector in codebook $\mcC_{\rm sk}$ is normalized to satisfy the unit norm constraint. It can be observed that the skewed codebook is design to match the covariance matrix of $\tilde \bfGamma$. The best codeword is selected as the one maximizing $|\hat{\bfGamma}_i^{\rmH}\bfLambda^{\rmH} \tilde\bfGamma |$.
\section{Rate Loss Analysis} \label{Ratelossa}
In this section, we analyze the rate loss of our proposed scheme. We decouple the channel prediction- and quantization-error and derive an upper bound of the rate loss. We show that a smaller subspace dimension is favorable for quantization due to less coefficients. On the other hand, more subspace coefficients will reduce the prediction error for the case of $\mathcal{D}_{\rm{ub}}>1$. Therefore, there exists a tradeoff between quantization error and prediction error when selecting the subspace dimension at a given feedback rate. We develop a subspace dimension switching algorithm to find the best tradeoff such that the rate loss upper bound is minimized. We also characterize the scaling of the required number of feedback bits in order to decouple the rate loss due to imperfect quantization from the transmit power. 

%\subsection{SISO Interference Alignment with Imperfect CSI}
%Imperfect CSI results in residual interference, thus, IA conditions (\ref{cond1}) and (\ref{cond2}) can not be satisfied. Let us define the average loss in sum rate as $I R= \mathbb{E}[R_{\rm sum}^{\rm perfect}]-\mathbb{E}[R_{\rm sum}]$, where $R_{\rm sum}^{\rm perfect}$ is the sum rate achieved with perfect CSI and the vectors in (\ref{cond1})-(\ref{cond3}), and $R_{\rm sum}$ is the sum rate given imperfect precoding vectors $\hat{\bfv}_k^{i}[m]$ and decoding vectors $\hat{\bfu}_k^{i}[m]$. An upper bound of the average loss in sum rate $I R$ is given by \cite{Ayach2012a}
%The authors of \cite{Ayach2012a} provide an upper bound of the average loss in sum rate given imperfect precoders $\hat{\bfv}_k^{i}[m]$ and decoders $\hat{\bfu}_k^{i}[m]$
%\be
%I R < \sum_{k,i}{ \frac{1}{N} {\rm{log}}_{2} \left(1+ \frac{\mathbb{E}\left[{\mathcal{I}}_{k,i}^1+{\mathcal{I}}_{k,i}^2 \right]} {\sigma^2}  \right)} \label{rateloss}
%\ee
%\be
%I R < \sum_{k,i}{ \frac{1}{N} {\rm{log}}_{2} \left(1+ \frac{\mathbb{E}\left[{\mathcal{I}}_{i}^{k}\right]} {\sigma^2}  \right)} \,, \label{rateloss}
%\ee

%\subsection{CSI Quantization and Achievable Rate Analysis}
\subsection{Leakage Interference Due to Imperfect CSI}
Imperfect CSI results in residual interference, thus, IA conditions (\ref{cond1}) and (\ref{cond2}) can not be satisfied. Upon reception of the quantized subspace vector ${\hat \bfGamma}$, the quantized channel $\hat{\bfh}[m]$ can be constructed in the same way as (\ref{recstCh}). The precoding vector $\hat\bfv_k^{i}[m]$ and decoding vector $\hat\bfu_k^{i}[m]$ are calculated using $\hat{\bfw}_{k,\ell}[m]= \bfD_{N \times S} {\hat\bfh}_{k,\ell}[m] $ as the true channel, which results in
\begin{align}
\hat\bfu_k^{i}[m]^{\rmH} \hat\bfW_{k,k}[m] \hat\bfv_k^{j}[m]&=0, \quad \forall k,~\forall i\neq j \\ 
\hat\bfu_k^{i}[m]^{\rmH} \hat\bfW_{k,\ell}[m] \hat\bfv_\ell^{j}[m]&=0, \quad \forall k\neq \ell,~\forall i,~j \\ 
\left|\hat\bfu_k^{i}[m]^{\rmH} \hat\bfW_{k,k}[m] \hat\bfv_k^{i}[m]\right| &\geq c>0,  \quad \forall k,~i \,.
\end{align}
We modify the upper bound for the average loss in sum rate \cite{Ayach2012a} for a time-variant channel as
\be
\Delta R < \frac{1}{NT}{ \sum_{k,i}\sum_{m \in \mcT} {\rm{log}}_{2} \left(1+ {\mathbb{E}\left[{\mathcal{I}}_{k,i}^{1} [m]+{\mathcal{I}}_{k,i}^{2} [m]\right]}  \right)}. \label{rateloss}
\ee

We define $\hat{\bfb}_{k,\ell}^{i,j}[m] =\hat{\bfu}_k^{i}[m]^{*} \circ \hat{\bfv}_\ell^{j}[m]$ as the Hadamard product of the decoding vector $\hat{\bfu}_k^{i}[m]$ and precoding vector $\hat{\bfv}_\ell^{j}[m]$. The leakage interference in (\ref{Inter1}) and (\ref{Inter2}) can be rewritten as
\begin{align}
{\mathcal{I}}_{k,i}^{1} [m]&=\sum_{i\neq j} \frac{NP}{d_k} \left|   \bfw_{k,k}[m]^{\rmT} \hat{\bfb}_{k,k}^{i,j}[m] \right|^2, \label{Inter1a} {\rm ~and}\\ 
{\mathcal{I}}_{k,i}^2 [m]&=\sum_{ k\neq \ell}\sum_{j=1}^{d_\ell} \frac{NP}{d_\ell} \left| \bfw_{k,\ell}[m]^{\rmT} \hat{\bfb}_{k,\ell}^{i,j}[m] \right|^2 \label{Inter2a}.
\end{align}
%\begin{align}
%{\mathcal{I}}_{i}^{k}=\sum_{(i,k)\neq(j,\ell)} \frac{NP}{d_\ell} \left| \bfw_{k,\ell}[m]^{\rmH} \hat{\bfb}_{k,\ell}^{i,j}[m] \right|^2 \label{Inter2a}.
%\end{align}
We define the  predicted channel frequency response as
\begin{align}
	{\tilde {\bfw}}_{k,\ell}[m]= [\tilde{w}_{k,\ell}^1 [m],\ldots ,\tilde{w}_{k,\ell}^n [m]]^{\rmT}
\end{align} 
and the prediction error as $\tilde{\bfz}_{k,\ell}[m]= \bfw_{k,\ell}[m]-\tilde {\bfw}_{k,\ell}[m]$. The average power leakage of the inter-stream interference in (\ref{Inter1a}) can be upper bounded by 
%\bea
\begin{align}
&\mathbb{E}\left[{\mathcal{I}}_{k,i}^{1} \left[m\right]\right] \nonumber\\
&=\sum_{{i\neq j}} \frac{NP}{d_k} \mathbb{E}\left[  \left|  \bfw_{k,k}[m]^{\rmT} \hat{\bfb}_{k,k}^{i,j}[m]   \right|^2
\right]\\
%&= \sum_{(i,k)\neq(j,\ell)} \frac{NP}{d_k} \mathbb{E} \left[ \left| \left(\hat{\bfw}_{k,\ell}[m]^{\rmH} + \bfz_{k,\ell}[m]^{\rmH}\right) \hat{\bfb}_{k,\ell} ^{i,j}[m] \right|^2 \right]\\
&= \sum_{{i\neq j}} \frac{NP}{d_k} \mathbb{E} \left[  \left| \left( \tilde{\bfw}_{k,k}[m]^{\rmT}+\tilde{\bfz}_{k,k}[m]^{\rmT} \right) \hat{\bfb}_{k,k}^{i,j}[m] \right|^2 \right] \\
&=\sum_{{i\neq j}} \frac{NP}{d_k} \Bigg( \mathbb{E} \bigg[ \left| \tilde{\bfw}_{k,k}[m]^{\rmT} \hat{\bfb}_{k,k}^{i,j}[m] \right|^2 +
 \left| \tilde{\bfz}_{k,k}[m]^{\rmT} \hat{\bfb}_{k,k}^{i,j}[m] \right|^2  \nonumber\\
& \quad + 2{\rm Re} \left( \tilde{\bfw}_{k,k}[m]^{\rmT} \hat{\bfb}_{k,k}^{i,j}[m] \hat{\bfb}_{k,k}^{i,j}[m]^{\rmH} \tilde{\bfz}_{k,k}[m]^{*} \right) \bigg] \Bigg) \label{3terms}\\
&\approx \sum_{{i\neq j}} \Bigg( \underbrace{\frac{NP}{d_k} \mathbb{E} \left[ \left|  \tilde{\bfz}_{k,k}[m]^{\rmT} \hat{\bfb}_{k,k}^{i,j}[m] \right|^2\right]}_{ \tilde{I} _{k,k}^{i,j}[m]}  \nonumber \\
&\quad +\underbrace{ \frac{NP}{d_k} \mathbb{E} \left[ \left|  \tilde{\bfw}_{k,k}[m]^{\rmT} \hat{\bfb}_{k,k}^{i,j}[m]   \right|^2 \right]}_ { \hat{I} _{k,k}^{i,j}[m]}\Bigg) \label{2terms}
%&\leq \sum_{(i,k)\neq(j,\ell)} \frac{NP}{d_k} \mathbb{E}  \left\|  \bfz_{k,\ell}[m] \right\|^2 \\
%&= \sum_{(i,k)\neq(j,\ell)} \frac{NP}{d_k} \mathbb{E}_{\bfz}  \left\|  \bfz_{k,\ell}[m] \right\|^2
\end{align}
%\eea
%where (\ref{2terms}) is obtained due to the fact that the prediction error $\tilde{\bfz}_{k,\ell}[m]$ is independent of the channel frequency response $\tilde{\bfw}_{k,\ell}[m]$, i.e.~$\mathbb{E} \left[ {\rm Re} \left( \tilde{\bfw}_{k,\ell}[m]^{\rmH} \hat{\bfb}_{k,\ell}^{i,j}[m] \hat{\bfb}_{k,\ell}^{i,j}[m]^{\rmH} \tilde{\bfz}_{k,\ell}[m] \right) \right]=0$.
where (\ref{2terms}) is obtained by ignoring the term $\mathbb{E} \left[ {\rm Re} \left( \tilde{\bfw}_{k,\ell}[m]^{\rmH} \hat{\bfb}_{k,\ell}^{i,j}[m] \hat{\bfb}_{k,\ell}^{i,j}[m]^{\rmH} \tilde{\bfz}_{k,\ell}[m] \right) \right]$. An equality holds from (\ref{3terms}) to (\ref{2terms}) with the MMSE predictor (\ref{mmsepre}), where $\tilde{\bfz}_{k,\ell}[m]$ is zero-mean Gaussian and independent of $\tilde{\bfw}_{k,\ell}[m]$. However, for our reduced-rank predictor, we notice that an exact characterization of this term $\mathbb{E} \left[ {\rm Re} \left( \tilde{\bfw}_{k,\ell}[m]^{\rmH} \hat{\bfb}_{k,\ell}^{i,j}[m] \hat{\bfb}_{k,\ell}^{i,j}[m]^{\rmH} \tilde{\bfz}_{k,\ell}[m] \right) \right]$ is mathematically intractable. As discussed in Remark~\ref{predrelation}, the reduced-rank predictor is closely related to the MMSE predictor. We also found via simulation that this term is rather small. Similar to (\ref{2terms}), the inter-user interference term in (\ref{Inter2a}) can be upper bounded by 
\begin{align}
\mathbb{E}\left[{\mathcal{I}}_{k,i}^{2} \left[m\right]\right]\approx& \sum_{k \neq \ell}\sum_{j=1}^{d_\ell} \Bigg( \underbrace{\frac{NP}{d_\ell} \mathbb{E} \left[ \left|  \tilde{\bfz}_{k,\ell}[m]^{\rmT} \hat{\bfb}_{k,\ell}^{i,j}[m] \right|^2\right]}_{ \tilde{I} _{k,\ell}^{i,j}[m]}  \nonumber \\
& +\underbrace{ \frac{NP}{d_\ell} \mathbb{E} \left[ \left|  \tilde{\bfw}_{k,\ell}[m]^{\rmT} \hat{\bfb}_{k,\ell}^{i,j}[m]   \right|^2 \right]}_ { \hat{I} _{k,\ell}^{i,j}[m]}\Bigg).\label{2termsa}
\end{align}

In the following Sections \ref{leakduetoprediction} and \ref{leakduetoquan}, we will show that the first and second terms in (\ref{2terms}) and (\ref{2termsa}) are caused by the channel prediction error and the quantization error, respectively. 

\subsection{Leakage Interference Due to Channel Prediction Error}\label{leakduetoprediction}
Defining $[\hat{\bfq}_{k,k}^{i,j}[m]^{\rmT} ,\bfq_{k,k}^{i,j}[m]^{\rmT} ]^{\rmT}=\bfD_N^{\rmH}\hat{\bfb}_{k,k}^{i,j}[m]^{*}$, where $\hat{\bfq}_{k,k}^{i,j}[m] \in \mathbb{C}^{S \times 1}$ and $\bfq_{k,k}^{i,j}[m] \in \mathbb{C}^{(N-S) \times 1}$, and $[\tilde{\bfe}_{xk,k}[m]^{\rmT}, \mathbf{0}_{1 \times (N-S)} ]^{\rmT}=\bfD_N^{\rmH} \tilde{\bfz}_{k,k}[m]$, the first term of the inter-stream interference $\tilde{I} _{k,k}^{i,j}[m]$ in (\ref{2terms}) can be written as
%\begin{align}
%\tilde{I} _{k,k}^{i,j}[m]=&\frac{NP}{d_k} \mathbb{E} \left[ \left| \tilde{\bfz}_{k,k}[m]^{\rmT} \hat{\bfb}_{k,k}^{i,j}[m] \right|^2\right] \nonumber\\
%=&\frac{NP}{d_k} \mathbb{E} \left[ \left| \tilde{\bfe}_{k,k}[m]^{\rmH} \hat{\bfq}_{k,k}^{i,j}[m] \right|^2\right]\label{persevv}\\
%\leq& \frac{NP}{S d_k} \mathbb{E} \left\|\hat{\bfq}_{k,k}^{i,j}[m] \right\|^2 \mathbb{E} \left\|  \tilde{\bfe}_{k,k}[m] \right\|^2 \label{f2d}\\
%=& \frac{NP}{Sd_k} \mathbb{E} \left\|\hat{\bfq}_{k,k}^{i,j}[m] \right\|^2 \sum_{n=1}^{N} \mathbb{E} \left| \tilde{z}_{k,k}[m,n] \right|^2 \\
%%=& \frac{NP}{Sd_k} \mathbb{E} \left\|\hat{\bfq}_{k,k}^{i,j}[m] \right\|^2 \left( N \cdot {\rm MSE}\left[m,D\right]-(N-S) \bff[m]^{\rmT}\bfG^{-1}\bff[m] \right)  \label{boundmse}\\
%=& \frac{N^2P}{Sd_k} \mathbb{E} \left\|\hat{\bfq}_{k,k}^{i,j}[m] \right\|^2  \cdot {\rm MSE}\left[m,D,\frac{NP}{S}\right] \label{boundmse}\\
%=& \tilde{J} _{k,k}^{i,j}[m] \nonumber
%\end{align}

\begin{align}
\tilde{I} _{k,k}^{i,j}[m]=&\frac{NP}{d_k} \mathbb{E} \left[ \left| \tilde{\bfz}_{k,k}[m]^{\rmT} \hat{\bfb}_{k,k}^{i,j}[m] \right|^2\right] \nonumber\\
=&\frac{NP}{d_k} \mathbb{E} \left[ \left| \tilde{\bfe}_{k,k}[m]^{\rmH} \hat{\bfq}_{k,k}^{i,j}[m] \right|^2\right]\label{persevv}\\
\approx&\frac{NP}{d_k} \mathbb{E} \bigg[  \hat{\bfq}_{k,k}^{i,j}[m]^{\rmH} \mathbb{E} \left[\tilde{\bfe}_{k,k}[m]  \tilde{\bfe}_{k,k}[m]^{\rmH}\right] \hat{\bfq}_{k,k}^{i,j}[m]\bigg]\label{assindp}\\
%=& \frac{NP}{Sd_k} \mathbb{E} \left\|\hat{\bfq}_{k,k}^{i,j}[m] \right\|^2 \sum_{n=1}^{N} \mathbb{E} \left| \tilde{z}_{k,k}[m,n] \right|^2 \\
%=& \frac{NP}{Sd_k} \mathbb{E} \left\|\hat{\bfq}_{k,k}^{i,j}[m] \right\|^2 \left( N \cdot {\rm MSE}\left[m,D\right]-(N-S) \bff[m]^{\rmT}\bfG^{-1}\bff[m] \right)  \label{boundmse}\\
=& \frac{N^2P}{Sd_k} \mathbb{E} \left\|\hat{\bfq}_{k,k}^{i,j}[m] \right\|^2  \cdot {\rm MSE}\left[m,D,\frac{NP}{S}\right] \label{boundmse}\\
=& \tilde{J} _{k,k}^{i,j}[m] \nonumber
\end{align}
where (\ref{persevv}) is due to Parseval's theorem. In order for tractable results,  we arrive at (\ref{assindp}) by assuming the independence of $\hat{\bfq}_{k,\ell}^{i,j}[m]$ and $\tilde{\bfe}_{k,\ell}[m]$, which is the case for the MMSE predictor ($\hat{\bfq}_{k,\ell}^{i,j}[m]$ is a function of $\hat{\bfh}_{k,\ell}[m]$ and therefore $\tilde{\bfh}_{k,\ell}[m]$. The prediction error $\tilde{\bfe}_{k,\ell}[m]$ is independent of $\tilde{\bfh}_{k,\ell}[m]$). For our reduced-rank predictor, it still provides a good approximation due to the close relation between a reduced-rank predictor and the MMSE predictor (see discussion in Remark~\ref{predrelation}). Equation (\ref{boundmse}) follows from $\mathbb{E} \left[\tilde{\bfe}_{k,k}[m]  \tilde{\bfe}_{k,k}[m]^{\rmH}\right]=\frac{\mathbb{E}\|\tilde{\bfe}_{k,k}[m]\|^2}{S} \bfI_S$, where $\|\tilde{\bfe}_{k,k}[m]\|^2=\|\tilde{\bfz}_{k,k}[m]\|^2$. The MSE per subchannel $\mathbb{E}\big[|\tilde{z}_{k,k}\left[m,n\right]|^2\big],~\forall n \in\{1,\ldots,N\}$, is the sum of a square bias and a variance term \cite{Zemen2007} $
{\rm{MSE}}[m,D,\SNR]=\text{bias}^{2}[m,D]+\text{var}[m,D,\SNR] \label{mse} $
where the variance can be approximated by
\begin{align}
\text{var}[m,D,\SNR]=\frac{\bff[m]^{\rmT}\bfG^{-1}\bff[m]}{\SNR}. 
\end{align}
The square bias term is calculated as \cite{Zemen2007}
\begin{align}
&\text{bias}^{2}[m,D]\nonumber\\
&=\int\limits_{-\frac{1}{2}}^{\frac{1}{2}} \left| 1-\bff[m]^{\rmT} \bfG^{-1} \sum_{\ell \in \mcP} \bff[\ell] e^{-j2\pi \nu (m-\ell)} \right|^2 S_h(\nu){\rm d}\nu
\end{align}
where $S_h(\nu)$ denotes the actual power spectral density of the fading process. Due to the removal of the noise terms in (\ref{linear}), the noise variance is reduced by a factor of $N/S$, therefore resulting in an SNR of $\frac{NP}{S}$ in equation (\ref{boundmse}). 

A similar result can be obtained for the first term of the inter-user interference in (\ref{2termsa}), i.e.
\begin{align}
\tilde{I} _{k,\ell}^{i,j}[m] \approx \tilde{J} _{k,\ell}^{i,j}[m]= \frac{N^2P}{Sd_\ell} \mathbb{E} \left\|\hat{\bfq}_{k,\ell}^{i,j}[m] \right\|^2 {\rm MSE}\left[m,D,\frac{NP}{S}\right]. \label{boundmse1} 
\end{align}
%The prediction error increases with $T$. In order to bound the interference leakage power, we need to make (\ref{boundmse}) independent of $P$. In this work, we want to reduce the pilot overhead. Thus, we consider to optimize $T$ with fixed $M$. The following algorithm is proposed to find the optimal $T$, i.e.~ 
%\begin{algorithm}
%\caption{Finding the maximal $T$}
%\label{findT}
%\begin{algorithmic}
%\STATE $T=1$;
%\WHILE {${\rm MSE}\left[m,\mathcal{D}_{\rm opt}\right] \leq 1/k_{1}P$}    
%\STATE $T \leftarrow T+1$;
%\ENDWHILE.
%\end{algorithmic}
%\end{algorithm}
%where $k_{1}$ is a relaxation factor, which does not affect the scaling.
%%Denoting $\hat{\bfGamma}_{k,\ell}$ as the quantized version of $\tilde{\bfGamma}_{k,\ell}$, we can obtain the quantized channel impulse response $\hat{\bfh}_{k,\ell}[m]=\bfF_{k,\ell}[m]\hat{\bfGamma}_{k,\ell}$ and the quantized channel frequency response ${\hat{\bfw}}_{k,\ell}[m]= {\mathcal{F}_N}\{[\hat{{\bfh}}_{k,\ell}[m]^{\rmT}. \mathbf{0}_{1 \times (N-S)}]^{\rmT}\}$.

\subsection{Leakage Interference Due to Channel Quantization Error}\label{leakduetoquan}
To obtain tractable expressions, we restrict the subsequent analysis to a flat PDP, such that the vector to be quantized $\breve{\bfGamma}_{k,k}$ is isotropically distributed with uncorrelated entires. The interference leakage caused by the quantization error $\hat{I} _{k,k}^{i,j}[m]$ in (\ref{2terms}) can be rewritten as
\begin{align}
\hat{I} _{k,k}^{i,j}[m]=&\frac{NP}{d_k} \mathbb{E} \left[ \left|  \tilde{\bfw}_{k,k}[m]^{\rmT} \hat{\bfb}_{k,k}^{i,j}[m]   \right|^2 \right]  \nonumber\\
=&  \frac{NP}{d_k} \mathbb{E} \left[ \left|  \tilde{\bfh}_{k,k}[m]^{\rmH} \hat{\bfq}_{k,k}^{i,j}[m] \right|^2 \right] \\
%=& \frac{NP}{d_k} \mathbb{E} \left[ \left| \tilde{\bfGamma}_{k,k}^{\rmT} \Lambda_{k,k}^{\rmT} \bfF_{k,k}[m]^{\rmT}  \hat{\bfq}_{k,k}^{i,j}[m] \right|^2 \right]. \label{repf} \\
=& \frac{NP}{d_k} \mathbb{E} \left[ \left| \hat{\bfq}_{k,k}^{i,j}[m]^{\rmH} \bfF_{k,k}[m]  \bfLambda_{k,k} \breve{\bfGamma}_{k,k}    \right|^2 \right]. \label{repf}
\end{align}
Since $\hat{\bfGamma}_{k,k}$ is the quantized version of $\breve{\bfGamma}_{k,k}$ and $\left\|\hat{\bfGamma}_{k,k}\right\|=1$, from Parseval's theorem we have $\hat{\bfq}_{k,k}^{i,j}[m]^{\rmH} \bfF_{k,k}[m]  \bfLambda_{k,k} \hat{\bfGamma}_{k,k} =0$. We can define an orthonormal basis in $\mathbb{C}^{DS}$ as 
\be
\left\{\hat{\bfGamma}_{k,k}, \frac{\bfLambda_{k,k}  \bfF_{k,k}[m]^{\rmH}  \hat{\bfq}_{k,k}^{i,j}[m] }{\left\|  \bfLambda_{k,k}  \bfF_{k,k}[m]^{\rmH}  \hat{\bfa}_{k,k}^{i,j}[m]  \right\|},\bfa_1,\bfa_2,\ldots,\bfa_{DS-2} \right\}\,,
\ee
where $[\bfa_1,\bfa_2,\ldots,\bfa_{DS-2}]$ is an orthonormal basis of null$\left(\left[\hat{\bfGamma}_{k,k}, \frac{\bfLambda_{k,k}  \bfF_{k,k}[m]^{\rmH}  \hat{\bfq}_{k,k}^{i,j}[m] }{\left\|  \bfLambda_{k,k}  \bfF_{k,k}[m]^{\rmH}  \hat{\bfq}_{k,k}^{i,j}[m]  \right\|}\right]^{\rmH}\right)$. We can decompose $\breve{\bfGamma}_{k,k}$ into the above orthonormal basis, i.e.
\begin{align}
&\left\| \breve{\bfGamma}_{k,k} \right\|^2 \nonumber=\left| \hat{\bfGamma}_{k,k}^{\rmH} \breve{\bfGamma}_{k,k}   \right|^2+ \nonumber\\
&\left|  \frac{ \hat{\bfq}_{k,k}^{i,j}[m]^{\rmH} \bfF_{k,k}[m]  \bfLambda_{k,k}}{\left\|   \hat{\bfq}_{k,k}^{i,j}[m]^{\rmH} \bfF_{k,k}[m]  \bfLambda_{k,k}\right\|} \breve{\bfGamma}_{k,k} \right|^2+ 
\sum_{m=1}^{DS-2} \left| \bfa_{m}^{\rmH}\breve{\bfGamma}_{k,k} \right|^2 . \label{orth}
%\\&\geq \left| \hat{\bfGamma}_{k,k}^{\rmH} \tilde{\bfGamma}_{k,k}   \right|^2+ 
%\left|  \frac{ \hat{\bfq}_{k,k}^{i,j}[m]^{\rmH}\bfF_{k,k}[m]}{\|  \bfF_{k,k}[m]^{\rmH} \hat{\bfq}_{k,k}^{i,j}[m] \|} \tilde{\bfGamma}_{k,k} \right|^2 \label{ineq}
\end{align}
Inserting (\ref{orth}) into (\ref{repf}) yields
\begin{align}
&\frac{NP}{d_k} \mathbb{E} \left[ \left| \hat{\bfq}_{k,k}^{i,j}[m]^{\rmH} \bfF_{k,k}[m]  \bfLambda_{k,k} \breve{\bfGamma}_{k,k}    \right|^2 \right]\\
&=\frac{NP}{d_k} \mathbb{E} \Bigg[ \left\|   \hat{\bfq}_{k,k}^{i,j}[m]^{\rmH} \bfF_{k,k}[m]  \bfLambda_{k,k}\right\|^2 \nonumber\\
& \ \qquad \left( \left\|\breve{\bfGamma}_{k,k}\right\|^2- |\hat{\bfGamma}_{k,k}^{\rmH} \breve{\bfGamma}_{k,k}|^2 -\sum_{m=1}^{DS-2} \left| \bfd_{m}^{\rmH}\breve{\bfGamma}_{k,k} \right|^2 \right) \Bigg] \\
&=\frac{NP}{d_k(DS-1)} \mathbb{E}\left\|  \hat{\bfq}_{k,k}^{i,j}[m]^{\rmH} \bfF_{k,k}[m]  \bfLambda_{k,k} \right\|^2 \nonumber\\
&\qquad \qquad \qquad \qquad \qquad \times \mathbb{E}\left[ \left\|\breve{\bfGamma}_{k,k}\right\|^2- |\hat{\bfGamma}_{k,k}^{\rmH} \breve{\bfGamma}_{k,k}|^2 \right] \label{proj}\\
%=& \frac{NP}{d_k(DS-1)} \mathbb{E} \left\| \bfF_{k,k}[m]^{\rmH}\hat{\bfq}_{k,k}^{i,j}[m]\right\|^2  \mathbb{E} \left\|\tilde{\bfGamma}_{k,k}\right\|^2 \mathbb{E} \left[ 1-\left|\frac{\hat{\bfGamma}_{k,k}^{\rmH} \tilde{\bfGamma}_{k,k}}{\left\|\tilde{\bfGamma}_{k,k}\right\|}\right|^2 \right] \label{normdir}\\
&= \frac{NP}{d_k(DS-1)} \mathbb{E} \left\| 
 \hat{\bfq}_{k,k}^{i,j}[m]^{\rmH} \bfF_{k,k}[m]  \bfLambda_{k,k}\right\|^2 \mathbb{E} \left\|\breve{\bfGamma}_{k,k}\right\|^2 \nonumber \\
& \ \qquad \qquad \qquad \qquad \qquad \times \mathbb{E} \left[ d_{\rm c}^2 \left( \frac{\breve{\bfGamma}_{k,k}}{\left\|\breve{\bfGamma}_{k,k}\right\|},\hat{\bfGamma}_{k,k}  \right) \right] \ \label{qbound} \
\end{align}
where $d_{\rm c}(\bfx_1,\bfx_2)=\sqrt{(1-|\bfx_1^{\rmH} \bfx_2|^2)}$ is the chordal distance between two unit norm vectors $\bfx_1$ and $\bfx_2$. Equation (\ref{proj}) follows from the fact that the quantization error is isotropic in the nullspace of $\hat{\bfGamma}_{k,k}$ and therefore the average power of $\breve{\bfGamma}_{k,k}$ in each dimension of $\left\{\frac{\bfLambda_{k,k}  \bfF_{k,k}[m]^{\rmH}  \hat{\bfq}_{k,k}^{i,j}[m] }{\left\|  \bfLambda_{k,k}  \bfF_{k,k}[m]^{\rmH}  \hat{\bfq}_{k,k}^{i,j}[m]  \right\|},\bfd_1,\bfd_2,\ldots,\bfd_{DS-2} \right\}$ is equal. Equation (\ref{qbound}) follows from the independence of the norm and the angle of $\breve{\bfGamma}_{k,k}$. 

Equation (\ref{qbound}) shows that the leakage interference can be bounded by the chordal distance between the true and the quantized subspace coefficients.  The term $Q(N_{\rmd}) =\mathbb{E} \left[ d_{\rm c}^2 \left( \frac{\breve{\bfGamma}_{k,\ell}}{\|\breve{\bfGamma}_{k,\ell}\|},\hat{\bfGamma}_{k,\ell}  \right) \right]$ in (\ref{qbound}) is the expectation of the quantization error. As shown in \cite{Dai2008}, for quantizing a vector arbitrarily distributed on the Grassmannian manifold $\mathcal{G}_{DS,1}$ using RVQ, the second moment of the chordal distance using $N_{\rmd}$ quantization bits can be bounded as
\begin{align}
Q(N_{\rmd}) \leq \frac{\Gamma(\frac{1}{DS-1})}{DS-1}{(c 2^{N_{\rmd}})}^{-\frac{1}{DS-1}}\,,
\end{align}
where $\Gamma(\cdot)$ denotes the Gamma function.

%As a result, to minimize the degradation in sum rate as shown in (\ref{rateloss}), we need a quantization strategy that minimizes the chordal distance $d_{\rm c} \left( \frac{\tilde{\bfGamma}_{k,\ell}}{\|\tilde{\bfGamma}_{k,\ell}\|},\hat{\bfGamma}_{k,\ell}  \right)$.
%
%As shown in \cite{Conway1996} and \cite{Dai2008}, if the quantizer's codebook is generated using the Grassmannian line-packing approach, the maximum quantization error in terms of chordal distance is upper bounded by
%\bea
%I_{d}&=& \max_{\bfA \in \mathbb{C}^{DS \times 1}, \|\bfA\|=1} d_{\rm c} \left( \bfA,\hat{\bfGamma}_{k,\ell}  \right), \quad \forall \bfA\\
%&\leq& \frac{2}{2^{\frac{N_{d}}{2(DS-1)}}}
%\eea

Furthermore, we have 
\begin{align}
% %\mathbb{E}\|\hat{\bfq}_{k,k}^{i,j}[m]^{\rmT} \bfF_{k,k}[m]  \bfLambda_{k,k}\|^2=\frac{\|\bfF_{k,\ell}[m]\|_{\rmF}^2}{N^2}
\left\|\hat{\bfq}_{k,k}^{i,j}[m]^{\rmH} \bfF_{k,k}[m]  \bfLambda_{k,k}\right\|^2&=\sum_{s=1}^{S}\left\| \hat{\bfq}_{k,k}^{i,j}[m,s] \bflambda_{k,k}\bff_{k,k}[m]  \right\|^2\nonumber\\
&= \left\|\hat{\bfq}_{k,k}^{i,j}[m]\right\|^2 \| \bflambda_{k,k}\bff_{k,k}[m] \|^2
\end{align}
and %$\mathbb{E}\|\tilde{\bfGamma}_{k,k}\|^2=DS $
\begin{align}
%%\mathbb{E}\|\tilde{\bfGamma}_{k,\ell}\|^2=\tr(\bfU^{\rmH}(\bfR_{h}+\sigma_{n}^2\bfI_L)\bfU), \\
%%\mathbb{E}\|\tilde{\bfGamma}_{k,k}\|^2=DN-(N-S){\tr}({\bfG_{k,k}^{-1} \bflambda_{k,k}^{-1}}) \label{expeeta}, 
\mathbb{E}\|\tilde{\bfGamma}_{k,k}\|^2=DS \label{expeeta}, 
\end{align}
according to Appendix \ref{app0}. Plugging the above results into (\ref{qbound}), the inter-stream interference leakage caused by quantization error $\hat{I} _{k,k}^{i,j}[m]$ can be finally bounded by
\begin{align}
\hat{I} _{k,k}^{i,j}[m] \leq & \hat{J} _{k,k}^{i,j}[m] \nonumber\\
 = &\frac{NPDS}{d_k  (DS-1)} \mathbb{E}\left\|\hat{\bfq}_{k,k}^{i,j}[m] \right\|^2 \left\| \bflambda_{k,k} \bff_{k,k}[m] \right\|^2 Q(N_{\rmd}). \label{boundNd}
\end{align}
Accordingly, the inter-user interference can be bounded as
\begin{align}
\hat{I} _{k,\ell}^{i,j}[m] \leq &  \hat{J} _{k,\ell}^{i,j}[m]\nonumber\\
  =& \frac{NPDS}{d_\ell  (DS-1)} \mathbb{E}\left\|\hat{\bfq}_{k,\ell}^{i,j}[m] \right\|^2 \left\| \bflambda_{k,\ell} \bff_{k,\ell}[m]\right\|^2 Q(N_{\rmd}). \label{boundNd1}
\end{align}
The only stochastic part in the equation is $\left\|\hat{\bfq}_{k,\ell}^{i,j}[m] \right\|^2$, whose value relies on the applied IA algorithm. 

%\be
%\frac{4}{2^{\frac{N_{d}}{DS-1}}} \propto \frac{1}{P} \\\label{ndscaling}
%\Rightarrow N_{d}= (DS-1) \log_{2}k P
%\ee

%\begin{theorem} \label{themdeq1}
%The average rate loss due to channel prediction and quantization can be upper bound by
%\begin{align}
%\Delta R <  \frac{1}{NT} \sum_{t \in \mcT}\sum_{k,i}{ {\rm{log}}_{2} \left(1+ \mathcal{J}_{k,i}^{1} [m]+\mathcal{J}_{k,i}^{2} [m]  \right)}, \label{rloss}
%\end{align}
%where $\mathcal{J}_{k,i}^{1}= \sum_{k \neq \ell} 
% \left(\tilde{J} _{k,k}^{i,j}[m]+\hat{J} _{k,k}^{i,j}[m] \right)$ and $\mathcal{J}_{k,i}^{2}=\sum_{i \neq j}\sum_{k=1}^{d_\ell}\left(\tilde{J} _{k,\ell}^{i,j}[m]+\hat{J} _{k,\ell}^{i,j}[m] \right)$.
%\end{theorem} 

\begin{theorem} \label{themdeq1}
When the proposed prediction and limited feedback strategy is used for IA CSI feedback, the average rate loss due to channel prediction and quantization can be upper bound by
\begin{align}
\Delta R \lesssim \Delta R_{\rm ub}&=
%&= \frac{1}{NT} \sum_{m \in \mcT}\sum_{k,i}{ {\rm{log}}_{2} \left(1+\mathcal{J}_{k,i}^{1} [m]+\mathcal{J}_{k,i}^{2} [m]  \right)} \label{rloss11}\\
 \frac{1}{NT}\sum_{k}\sum_{m \in \mcT}  {d_k}   {\rm{log}}_{2} \Bigg(1+ NP \left(K-\frac{1}{d_k}\right) \cdot \nonumber \\
&  \bigg(\frac{N}{S} {\rm MSE}\left[m,D,\frac{NP}{S}\right] + \frac{  DS \zeta[m] Q(N_{\rmd})}{DS-1} \bigg)  \Bigg),  \label{finalbound}
\end{align}
where $\zeta[m]=\left\| \bflambda_{k,\ell} \bff_{k,\ell}[m] \right\|^2$.
\end{theorem} 
\begin{proof}
Equation (\ref{finalbound}) is obtained by inserting (\ref{boundmse}), (\ref{boundmse1}), (\ref{boundNd}) and (\ref{boundNd1}) into (\ref{rateloss})	and using the fact $\left\|\hat{\bfq}_{k,\ell}^{i,j}[m] \right\|^2<1$, $\forall (i,k,j,\ell)$. This can be shown as $\left\|\hat{\bfq}_{k,\ell}^{i,j}[m] \right\|^2 \leq \left\|\hat{\bfb}_{k,\ell}^{i,j}[m] \right\|^2=\sum_{n=1}^N \left|\hat{u}_{i}^{k}[m,n]\right|^2  \left|\hat{v}_{j}^{\ell}[m,n]\right|^2<\sum_{n=1}^N \left|\hat{u}_{i}^{k}[m,n]\right|^2  \sum_{n=1}^N \left|\hat{v}_{j}^{\ell}[m,n]\right|^2=1$.
\end{proof}

\begin{remark}  \label{themdeq1rmk}
We notice that the rate loss upper bound (\ref{finalbound}) derived using the method \cite{Ayach2012a} is known to be loose especially when CSI quality is poor, mainly due to the use of Jensen's inequality. Besides, we use the fact $\left\|\hat{\bfq}_{k,\ell}^{i,j}[m] \right\|^2<1$, which further loosens the bound. However, the term $\left\|\hat{\bfq}_{k,\ell}^{i,j}[m] \right\|^2$ exists in both the prediction and quantization errors, thus using this inequality is not critical for the purpose of deriving a subspace switching algorithm in Section \ref{adptsds}, especially at high SNRs.
\end{remark} 
\begin{theorem} \label{themdeq2}
The sum rate loss due to the quantization error can be bounded by a finite value when $P\longrightarrow\infty$, if the number of feedback bits per receiver grows as
	\begin{align}
	N_{\rm d}= (DS-1) \log_{2} P. \label{ndscaling}
	\end{align}
\end{theorem} 
\begin{proof}
	See Appendix \ref{app2}.
\end{proof}

\zhinan{  
Note that practical implementation of large codebooks in order to achieve a low quantization error remains a long-standing problem and is still under investigation. Ongoing research topics that are helpful in complexity reduction includes progressive refinement \cite{Heath2009a} and hierarchical codebooks \cite{Boccardi2007}.}

\subsection{Adaptive Subspace Dimension Switching Algorithm} \label{adptsds}
When the subspace coefficients are unquantized, the optimal subspace dimension that minimizes the prediction error is given by (\ref{optD}). However, for $\mathcal{D}_{\rm{ub}}>1$, a subspace dimension higher than one is favorable for channel prediction, while resulting in a higher quantization error. Hence, a limited feedback system exhibits a tradeoff between the quality of channel prediction and quantization. The selection of the subspace dimension to find the best tradeoff becomes more relevant and thus, a selection metric is needed for this purpose. The rate loss upper bound developed in (\ref{finalbound}) is suitable. We propose an {\em adaptive subspace dimension switching} algorithm, which finds the subspace dimension minimizing (\ref{finalbound}), i.e.~
\begin{align}
%\mathcal{D}_{\rm{opt}}=\underset{D \in \{1,\ldots,\mathcal{D}_{\rm ub}\}}{\arg \min} \frac{1}{NT} \sum_{t \in \mcT}\sum_{k,i}{ {\rm{log}}_{2} \left(1+ \mathcal{J}_{k,i}^{1} [m]+\mathcal{J}_{k,i}^{2} [m]   \right)}. \label{tradeoff}
\mathcal{D}=\underset{D \in \{1,\ldots,\mathcal{D}_{\rm ub}\}}{\arg \min} {\Delta R_{\rm ub}  }. \label{tradeoff}
\end{align}

%Unlike in MIMO IA case where the channel matrix is standard Gaussian matrix, the bi-unitarily invariant nature of Gaussian matrix gives independent precoders and decoders. 
%In SISO IA case, $\left\|\hat{\bfq}_{k,\ell}^{i,j}[m] \right\|^2$ is dependent on the IA applied algorithm, we consider a high SNR approximation and assume the value of $\mathbb{E} \left\|\hat{\bfq}_{k,\ell}^{i,j}[m] \right\|^2$ remain unchanged $\forall (i,k) \neq (j,\ell)$,
%
%\begin{align}
%&\Delta R < \Delta R_{\rm ub}\nonumber\\
%&= \frac{1}{NT} \sum_{t \in \mcT}\sum_{k,i}{ {\rm{log}}_{2} \left(1+\mathcal{J}_{k,i}^{1} [m]+\mathcal{J}_{k,i}^{2} [m]  \right)} \label{rloss11}\\
%%&=\frac{\xi}{NT}  \sum_{t \in \mcT}{ {\rm{log}}_{2} (\xi-1)\left(\mathcal{J}_{k,i}^{1} [m]+\mathcal{J}_{k,i}^{2} [m]  \right)} \\
%%&= \frac{1}{NT} \sum_{t \in \mcT}\sum_{k,i} {\rm{log}}_{2} \bigg(\frac{NP}{d_\ell}  {\rm MSE}\left[m,D\right] + \nonumber\\
%%& \big( \rho[m] \left(DN-(N-S)\zeta\right) Q(N_{\rmd})\big)  \bigg) \\
%&\leq \sum_{i}\sum_{t \in \mcT}  \frac{d_k}{NT}    {\rm{log}}_{2} \bigg(1+ NP \left(K-\frac{1}{d_k}\right) \nonumber \\
%&  \left( {\rm MSE}\left[m,D\right] + \nonumber \left( DN-\zeta(N-S)\right)\rho[m] Q(N_{\rmd}) \right)  \bigg)
%\end{align}
%
%where $\zeta={\tr}\left({\bfG_{k,\ell}^{-1} \bflambda_{k,\ell}^{-1}}\right)$ and $\rho[m]=\left\|\bff_{k,\ell}[m] \circ \bflambda_{k,\ell}\right\|^2$
%

\section {Simulation Results} \label{numersim}
In this section, the sum rate of the proposed scheme is evaluated through Monte-Carlo (MC) simulations. For the IA design in this section, we use the closed-form IA algorithm \cite{Cadambe2008a}  over $N=5$ channel extensions with an additional precoding subspace optimization \cite{Fadlallah2014}, since it has been shown that the original closed-from IA solution \cite{Cadambe2008a} yields low rate if no further optimization is performed \cite{Fadlallah2014,Kim2010}.
%The effective sum rate is defined as $R_{\rm eff}=\frac{T}{M+T}R_{\rm sum}$.
We consider a $K=3$ user interference channel, where each channel has $S$ delay taps and a flat PDP $\mathbb{E}\{\bfh_{k,\ell}[m]\bfh_{k,\ell}[m]^{\rmH}\}=\frac{N}{S}{\bfI_S}$. Each delay tap $h_{k,\ell}[m,s]$ is temporally correlated according to Clarke's model \cite{Clarke1968} with $R_{\bfh_{k,\ell}}[m]={{J}}_0(2\pi\nu_{\rm{D}}m)$, where ${J}_0$ is the $0$-th order Bessel function of the first kind. The OFDM symbol rate $1/T_s=1.4 \times 10^4$Hz is chosen according to the 3GPP LTE standard\cite{ieeephy}. The carrier frequency $f_c=2.5{\rm GHz}$. \zhinan{The normalized Doppler frequency is obtained as $\nu_{\rm D}= v f_c T_s/ c_0$, where $c_0$ is the speed of light, and $v$ is the relative velocity between transmitter and receiver.} In order to enable the performance analysis with exponentially large codebooks,  we replace the RVQ process by the statistical model of the quantization error using random perturbations \cite[Sec.~VI.B]{Rezaee2012}, which has been shown to be a good approximation of the quantization error using RVQ.

%We consider a channel with $S$ uncorrelated delay taps and a power delay profile $\mathbb{E}\{\bfh_{k,\ell}[m]\bfh_{k,\ell}[m]^{\rmH}\}={\bfI_L}/L $.
%Fig.~\ref{scaling_rate} illustrates the sum rate of the three user interference channel over three frequency extensions with feedback delay. The precoders are calculated using the close-form IA algorithm \cite{Cadambe2008a} over N=3 channel extensions, which corresponds to a $dof=4/3$. A feedback delay of $T_{\rm D}=1$ms, i.e.~the duration of 14 OFDM symbols, is considered $\forall k,\ell$. We further assume that the delays are known at all transmitters. Thus, the proposed scheme allow the transmitter to predict the channel in the transmission interval $\mcT=[M+15,\ldots,M+14+T]$.  At each SNR, the maximal $T$ is found according to algorithm \ref{findT} with $k_{1}=0.1$.

\subsection {Validation of The Rate Analysis}
First, we examine the effect of imperfect channel prediction and quantization. Fig.~\ref{ialeakage} shows the power of leakage interference (for a specific $(k,\ell)$ and $(i,j)$) versus the evolution of time for $\nu_{\rmD}=0.001$ (6.05 km/h). %\diag\left([{\frac{1}{2}, \frac{1}{3},\frac{1}{6}}]\right)$. 
The leakage powers due to prediction error and quantization error are shown respectively for MC simulations of $\tilde{I} _{k,\ell}^{i,j}[m]$ and $\hat{I} _{k,\ell}^{i,j}[m]$, and for the analytical upper bound $\tilde{J} _{k,\ell}^{i,j}[m]$ and $\hat{J} _{k,\ell}^{i,j}[m]$. Note that there still exists a stochastic part $ \mathbb{E}\left\|\hat{\bfq}_{k,\ell}^{i,j}[m] \right\|^2$ in $\hat{J} _{k,\ell}^{i,j}[m]$ and $\tilde{J} _{k,\ell}^{i,j}[m]$. However, as explained in Remark \ref{themdeq1rmk}, using the upper bound $\left\|\hat{\bfq}_{k,\ell}^{i,j}[m] \right\|^2<1$ has only a minor impact on the subspace switching algorithm, and thus the deterministic part of $\hat{J} _{k,\ell}^{i,j}[m]$ and $\tilde{J} _{k,\ell}^{i,j}[m]$ are more relevant. Therefore, we take an empirical value of $ \mathbb{E}\left\|\hat{\bfq}_{k,\ell}^{i,j}[m] \right\|^2$ from the simulation in order to make the comparison with the true leakage power.  We can observe that the leakage due to the prediction error increases over time due to increased MSE. The leakage due to the quantization error is almost a constant throughout the frame. In addition, the results corresponding to MC simulation and the analytical upper bound are quite close. The sum of both leakage terms is slightly higher than the true interference leakage power due to the ignorance of the last term in (\ref{3terms}). \zhinan{For comparison, we include also the prediction error without noise reduction as described in Sec.~\ref{noiosereduction}. It can be seen that the prediction error is larger due to the higher noise level.} Note that the interference leakage with non-flat PDPs (not shown) is similar to the one with flat PDP and matches well with the analytical bound as well.  
\begin{figure}[t]
	\centering
	\includegraphics[width=1\columnwidth]	{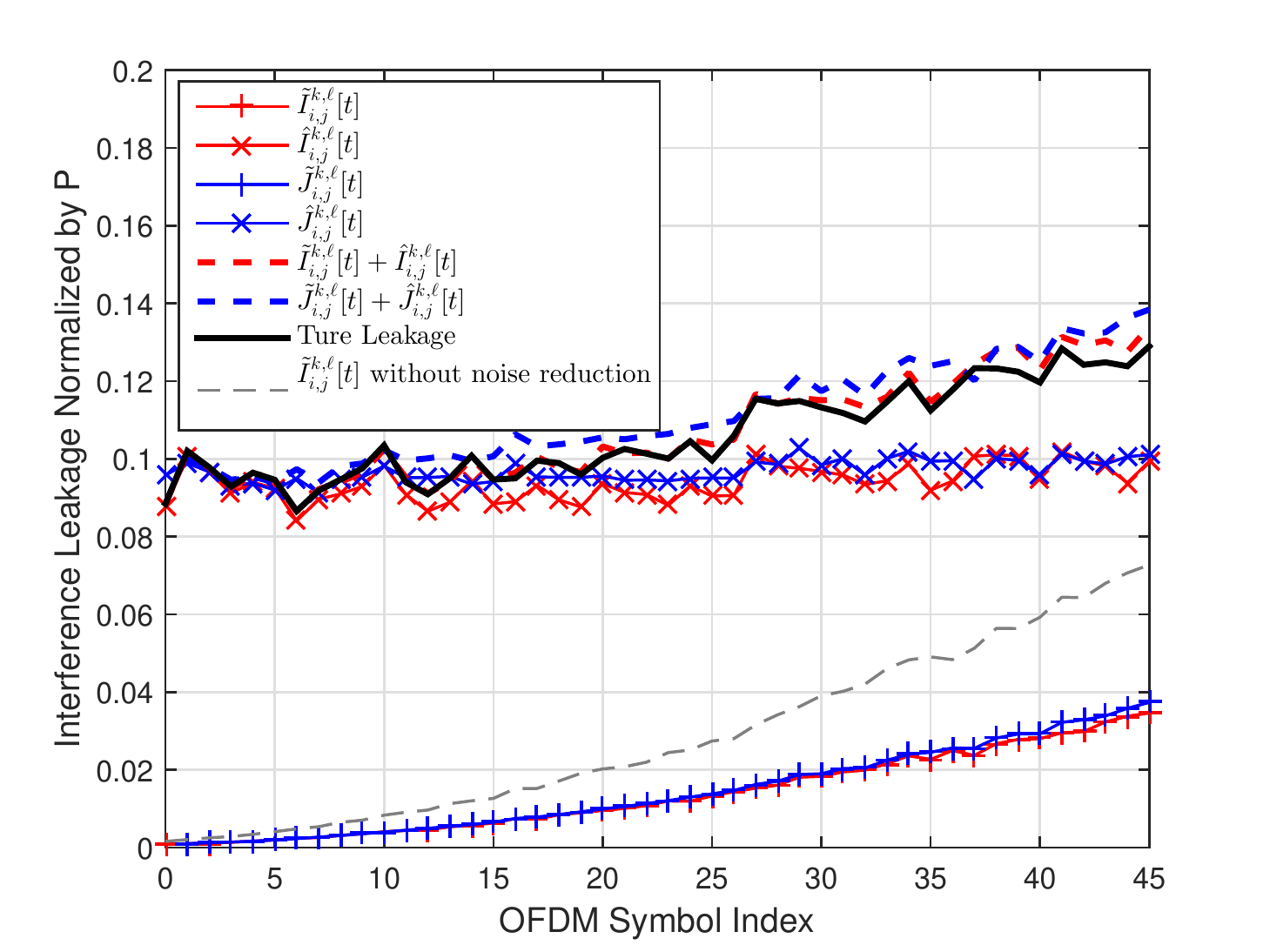}
	\caption{Evolution of Interference leakage with time at SNR=$25$dB and normalized Doppler frequency $\nu_{\rm D}=0.001$. The length of the pilot sequence $M=15$. The length of the payload $T=45$. The number of channel taps $S=3$. The number of symbol extensions $N=5$. The number of feedback bits $N_{\rmd}=15$.} 
	\label{ialeakage}
\end{figure}

\begin{figure}[htbp]
	\centering
	\includegraphics[width=1\columnwidth]	{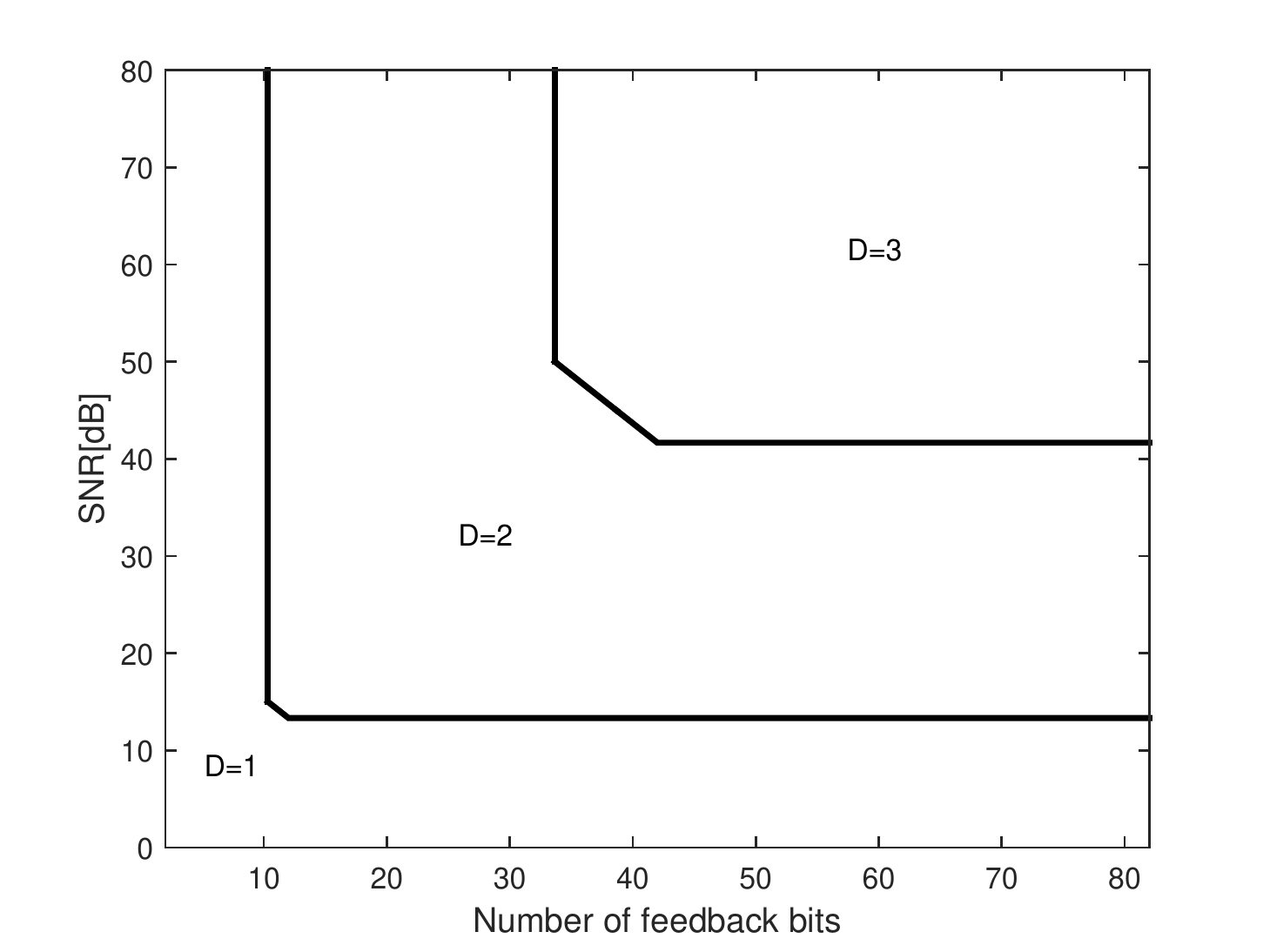}
	\caption{Subspace dimension obtained according to (\ref{tradeoff}), as a function of SNR and the number of feedback bits at normalized Doppler frequency $\nu_{\rm D}=0.004$. The length of the pilot sequence $M=15$. The length of the payload $T=45$. The number of channel taps $S=2$. The number of symbol extensions $N=5$.} 
	\label{subspaceDcontour}
\end{figure}

\begin{figure}[htbp]
	\centering
	\includegraphics[width=1\columnwidth]	{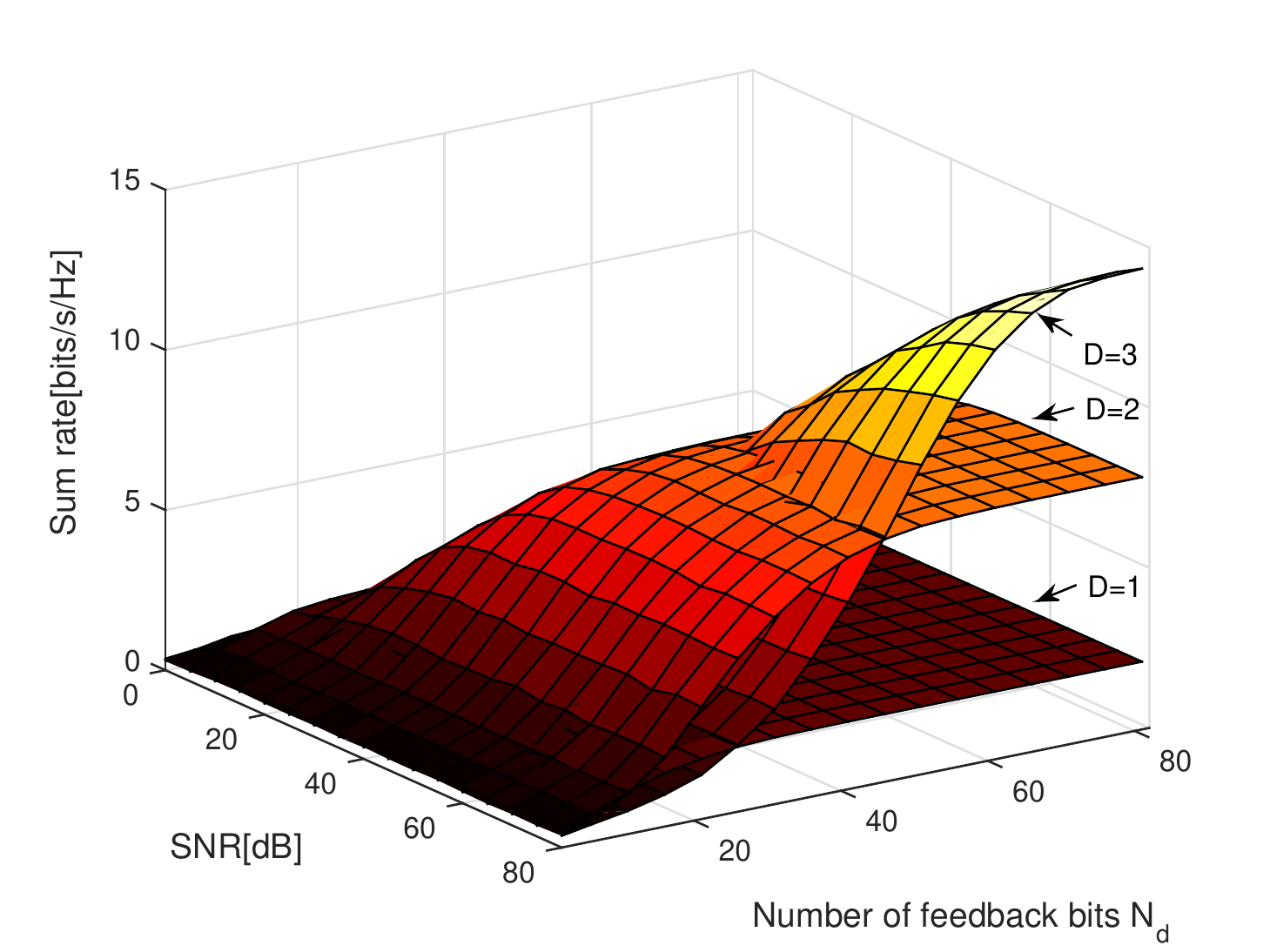}
	\caption{Sum rate with subspace dimension $D \in \{1,2,3\}$, as a function of SNR and the number of feedback bits at normalized Doppler frequency $\nu_{\rm D}=0.004$. The length of the pilot sequence $M=15$. The length of the payload $T=45$. The number of channel taps $S=2$. The number of symbol extensions $N=5$.} 
	\label{subspaceDrate}
\end{figure}

\subsection {Choice of Subspace Dimension}
\begin{figure}[t]
	\centering
	\includegraphics[width=1\columnwidth]	{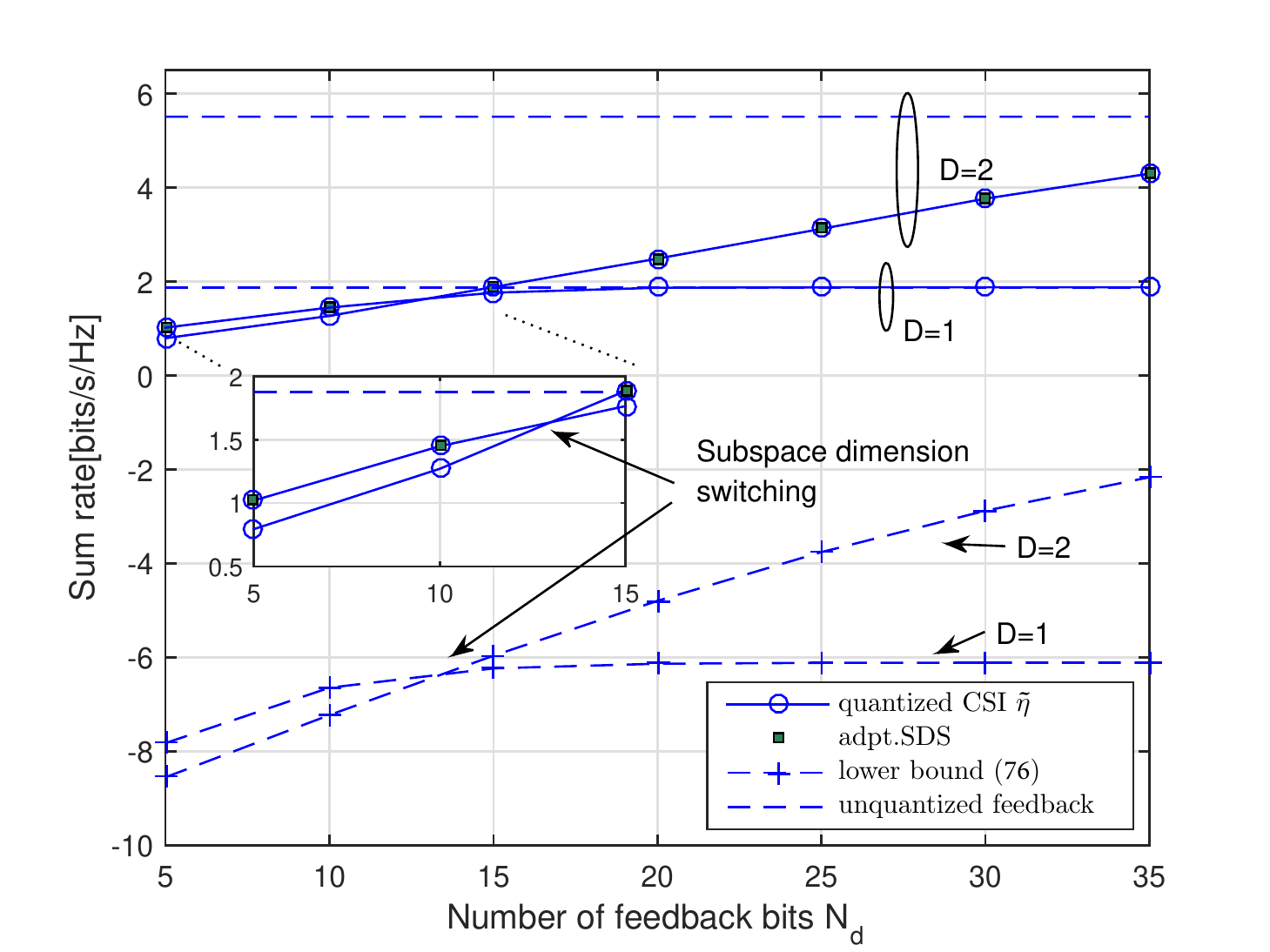}
	\caption{Sum rate versus the number of feedback bits at SNR=~30dB and the normalized Doppler frequency $\nu_{\rm D}=0.004$. The length of the pilot sequence $M=15$. The length of the payload $T=45$. The number of channel taps $S=3$. The number of symbol extensions $N=5$.} 
	\label{fig2}
\end{figure}

\zhinan{In this section, we demonstrate the subspace choice and the accuracy of the subspace dimension switching algorithm with various SNRs, number of feedback bits and number of delay taps. }

Fig.~\ref{subspaceDcontour} shows contour lines of the subspace dimension obtained according to (\ref{tradeoff}) as a function of SNR and the number of feedback bits. \zhinan{In order to demonstrate the subspace dimension switching algorithm, we select $\nu_{\rm D}=0.004$ (24.2km/h)  for the simulation, which represents the velocity of a slowly moving car.}
%The subspace dimension is obtained according to (\ref{tradeoff}), and compared to the optimal subspace dimension by  MC simulations
%\begin{align}%
%D_{\rm opt}= \arg \max_{D=\{1,\ldots,D_{\rm ub}\}} \mathbb{E}[R_{\rm sum}].
%\end{align} 
It can be seen that a higher subspace dimension is suggested when both SNR and the number of feedback bits are high. This is because higher SNR allows for high subspace dimension for channel prediction due to the relatively small variance of a reduced-rank predictor. This will also result in more subspace coefficients, which in turn require more bits for feedback to maintain a low quantization error. In case of a low feedback rate, a lower subspace dimension is still favorable in order for a low quantization error, and therefore the best tradeoff between prediction and quantization.

Fig.~\ref{subspaceDrate} illustrates the sum rate using the same setup as in Fig.~\ref{subspaceDcontour}, with subspace dimension $D\in\{1,2,3\}$, respectively. It can be seen that the dimension suggested in Fig.~\ref{subspaceDcontour} matches well with the dimension that achieves a higher rate.

\begin{figure}[t]
	\centering
	\includegraphics[width=1\columnwidth]{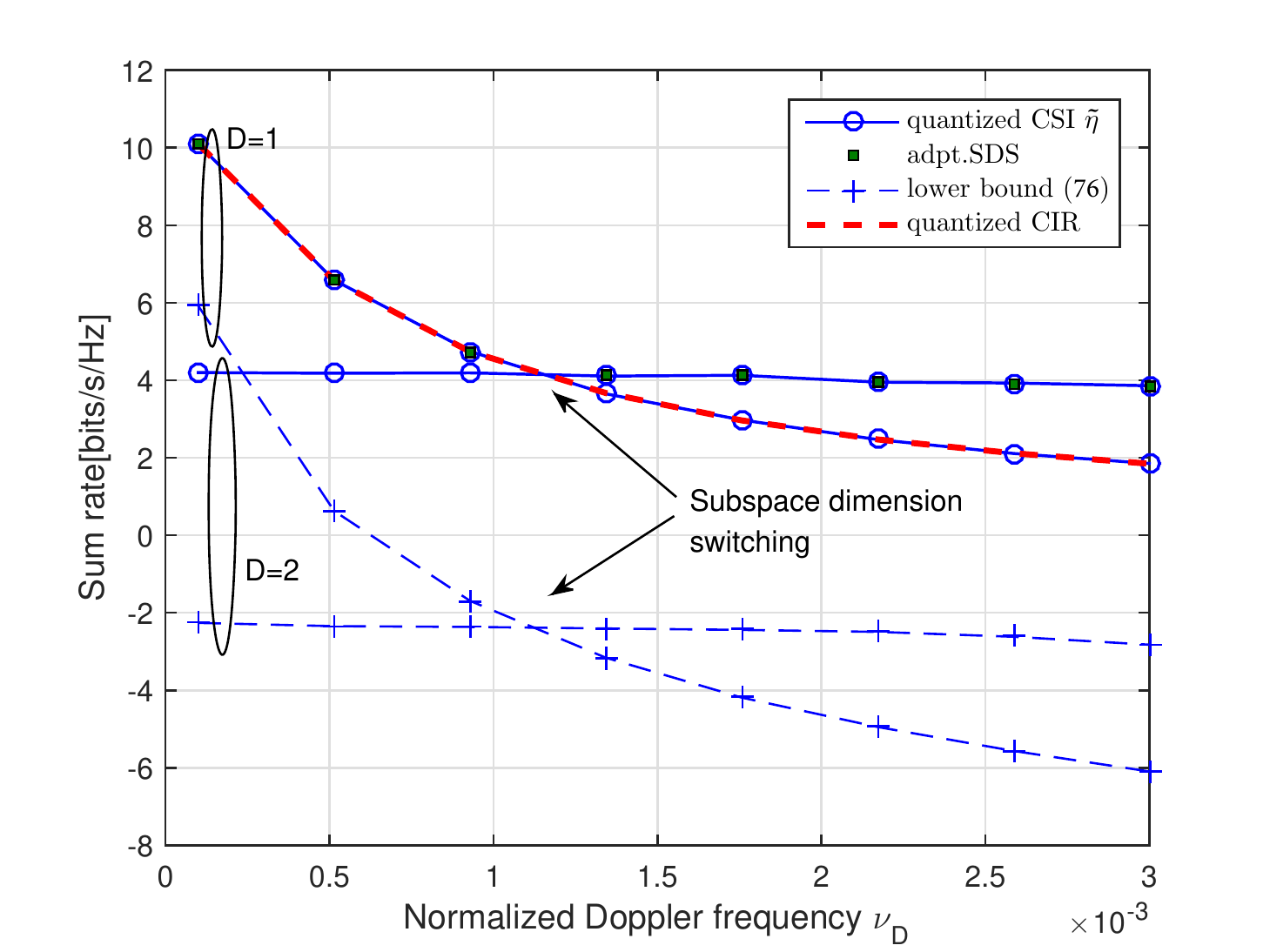}
	\caption{Sum rate degradation versus the increase of the normalized Doppler frequency. A 0.5 ms feedback delay $T_{\rm D}=7$ is considered $\forall k,\ell$. The length of the pilot sequence $M=15$. The length of the payload $T=45$. The number of feedback bits $N_{\rm d}=30$. The number of channel taps $S=3$. The number of symbol extensions $N=5$.}
	\label{rate_doppler}
\end{figure}
Fig.~\ref{fig2} shows the sum rate versus the number of feedback bits at an SNR=30dB and the normalized Doppler frequency $\nu_D =0.004$ (24.2km/h).  The lower bound of the average achievable rate is defined as 
\begin{align}
R_{\rm lb}=\mathbb{E}\left[ R_{\rm sum}^{\rm perfect} \right]-\Delta R_{\rm ub} .
\end{align}
Due to the fact that the average sum rate given perfect CSI is a constant, we can use this lower bound to examine the effectiveness of the subspace switching algorithm ($\ref{tradeoff}$). \zhinan{Note that the rate lower bound can be negative due to the looseness of the rate loss upper bound $\Delta R_{\rm ub}$, as discussed in Remark \ref{themdeq1rmk}.} For such a setting, (\ref{optD}) suggests that the optimal subspace dimension $\mathcal{D}_{\rm ub}$ is 2 for unquantized feedback. \zhinann{However, as discussed earlier, higher dimension $D$ will lead to a larger quantization error due to more subspace coefficients.} To find the best subspace dimension, we present the achieved rate and the corresponding lower bound at both $D=\{1,2\}$. It can be observed that the achieved sum rate increases with the number of feedback bits. For $D=1$, it achieves an initial higher rate due to smaller quantization error. The achieved rate becomes a constant with the increase of $N_d$ due to the dominance of the prediction error. When more than 15 bits are used, the two dimensional subspace outperforms the one dimensional subspace due to the better capability of channel prediction. The tradeoff between the quality of channel prediction and quantization is well captured by the lower bounds, which exhibit almost the same switching point as that obtained by MC simulation. Thus, the adaptive subspace dimension switching algorithm (\ref{tradeoff}), denoted by adpt.SDS, is efficient to find the subspace dimension associated with a higher rate.

Fig.~\ref{rate_doppler} shows the sum rate degradation as the increase of the normalized Doppler frequency with a feedback delay $T_{\rm D}=7$ (0.5 ms) $\forall k,\ell$. \zhinan{Due to this additional feedback delay, we reduce the payload size to $T=30$ OFDM symbols to obtain a good channel prediciton.}
The performance is also compared to the traditional non-predictive strategy \zhinann{(represented as ``quantized CIR'' by the red dashed line)}, which feeds back the channel impulse response (CIR) and assumes the channel is constant over the frame length. The estimate of the impulse response is obtained using the solution presented in Sec.~\ref{RR} and then averaged over all pilot positions. At low Doppler frequency, a lower subspace dimension is selected. For $D=1$, the rates achieved by non-predictive and proposed algorithms are similar. This is due to the first dimensional DPS sequence is almost a constant, therefore incapable to predict the channel. As the Doppler frequency increases, the DPS sequences of dimension $D=2$ outperform when the rate increase due to better channel prediction is higher than the rate decrease due to increased quantization error. It also can be seen that the intersection point of the sum rate lower bound for $D \in \{1,2\}$ is almost the same as the one for the  MC simulation. Therefore, by evaluating the rate loss upper bound, the adpt.SDS algorithm (\ref{tradeoff}) is able to select the subspace dimension with a higher rate. 

\subsection {Numerical Results on Sum Rate}
\begin{figure}[t]
\centering
\includegraphics[width=1\columnwidth]{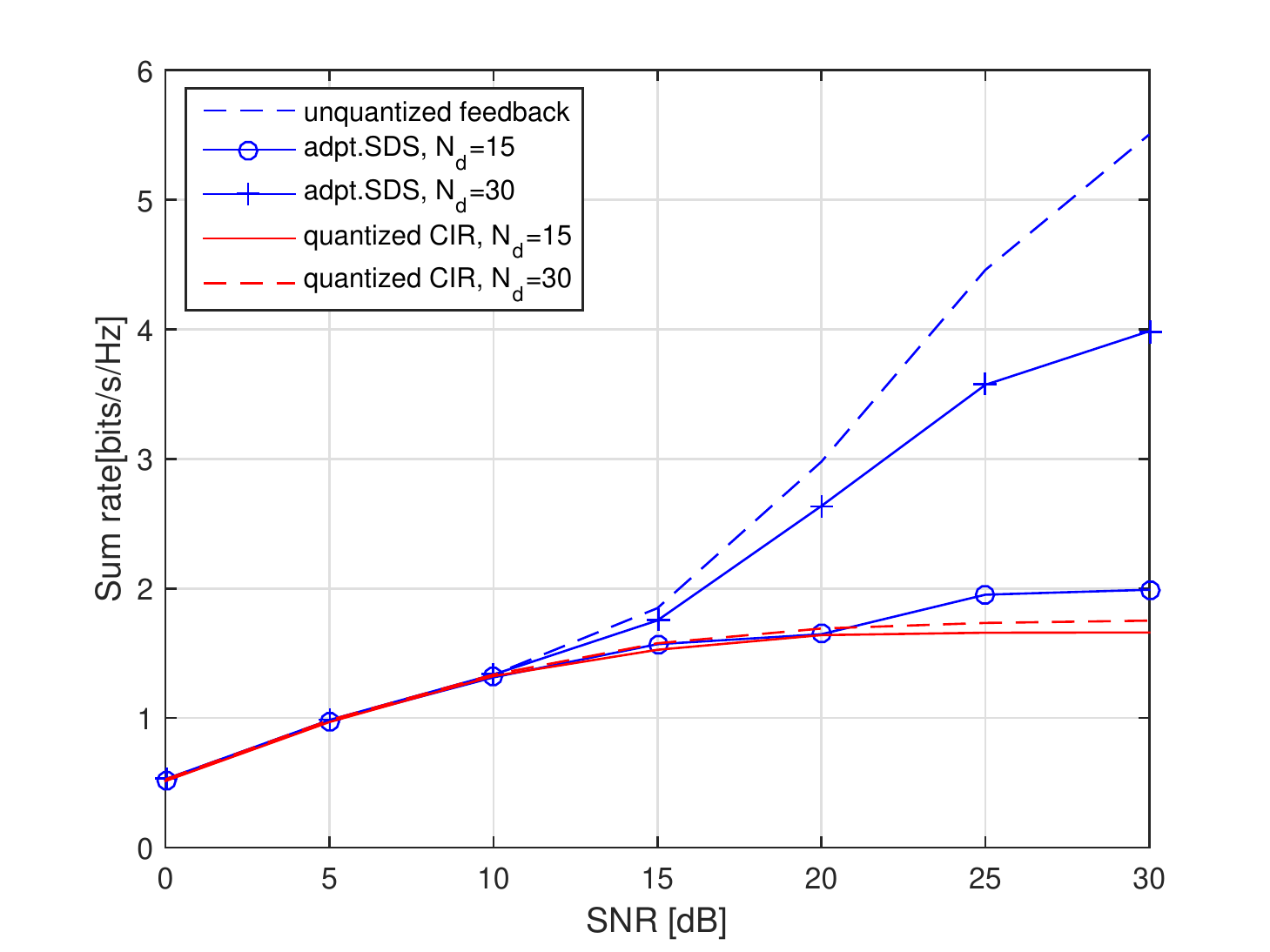}
\caption{Sum rate versus SNR at normalized Doppler frequency $\nu_{\rm D} =0.004$. A 0.5 ms feedback delay $T_{\rm D}=7$ is considered $\forall k,\ell$. The length of the pilot sequence $M=15$. The length of the payload $T=30$. The number of channel taps $S=3$. The number of symbol extensions $N=5$.}
\label{fig3}
\end{figure}
%
%We are concerned with  fixed pilot and payload length in this work. 
Fig.~\ref{fig3} illustrates the sum rate at normalized Doppler frequency $\nu_{\rm D} =0.004$ (24.2km/h) with feedback delay $T_{\rm D}=7$ (0.5 ms) $\forall k,\ell$. The prediction algorithm with adapt.SDS has a subspace dimension $D=1$ at low SNRs, which results in a similar performance to ``quantized CIR''. For $N_{\rm d}=30$, the optimal subspace dimension $D$ switches to 2 at $\SNR=15$~dB. For $N_{\rm d}=15$, the switch takes place later at $\SNR=20$~dB.  As a result, better channel prediction is achieved at higher SNR, especially for a large number of feedback bits. The adaptive subspace dimension switching algorithm is able to efficiently find the dimension associated with a higher rate, which guarantees the superiority of the proposed feedback scheme over the non-predictive strategy.

\section {\zhinan{Conclusion}} \label{Conclusion}
We proposed a novel limited feedback algorithm for SISO interference alignment. The feedback algorithm enables reduced-rank channel prediction, which reduces the channel estimation error due to user mobility and feedback delay. We derived an upper bound of the rate loss due to channel prediction and quantization error, which was used to facilitate an adaptive subspace dimension switching algorithm. The algorithm is efficient to choose the subspace dimension associated with a higher rate by tradeoff between prediction error and quantization error. We characterized the scaling of the required number of bits in order to decouple the rate loss due to channel quantization from the transmit power. Simulation results showed that a higher subspace dimension is preferred for high SNR regime with an adequate number of feedback bits. By adaptively choosing the dimension, a rate gain over the non-predictive strategy can be obtained. With moderate velocities, a rate gain of $60\%$ can be achieved at $20$dB SNR.

\zhinan{ %This paper demonstrated the advantage of feeding back subspace coefficients for channel prediction in interference channels. %More work will be needed on the estimation of second order statistics of the time-variant fading process, or the support of the Doppler spectrum to perform channel prediction in practice. 
%	More work will be needed on practical implementation of large codebooks in order to achieve a low quantization error. This remains a long-standing problem and is still under investigation. In addition, SISO interference alignment requires a frequency-selective fading environment, where the required number of subchannels increases fast with the number of users. Introducing dimensions in spatial domain will relieve this problem. Therefore, future research direction also includes the extension of the proposed algorithm to frequency-selective MIMO channels. Besides, for future work we would like to characterize the optimum pilot overhead by a joint optimization of the pilot length and the payload length, taking into account the number of users, capacity of feedback channels and SNR.
}

\begin{appendices} 

\section{Proof of (\ref{expeeta})} \label{app0}
The delay domain subspace vector $\tilde{\bfGamma}_{k,k}$ can be written as 
\begin{align}
&\mathbb{E}\|\tilde{\bfGamma}_{k,k}\|^2\nonumber\\
&=\sum_{s=1}^{S}\mathbb{E}\| \bflambda_{k,k}^{-\frac{1}{2}} \tilde{\bfgamma}_{k,k}^s\|^2\\
&=\sum_{s=1}^{S}\mathbb{E}\[ {\tr}\left({\bflambda_{k,k}^{-\frac{1}{2}} \tilde{\bfgamma}_{k,k}^s \tilde{\bfgamma}_{k,k}^{s \rmH} \bflambda_{k,k}^{-\frac{1}{2}}}\right)\] \\
&=  {\tr}\left({\bflambda_{k,k}^{-\frac{1}{2}}\sum_{s=1}^{S}\mathbb{E}\{ \tilde{\bfgamma}_{k,k}^s \tilde{\bfgamma}_{k,k}^{s \rmH}\} \bflambda_{k,k}^{-\frac{1}{2}}}\right)\\
%&=  {\tr}\left({\bflambda_{k,k}^{-\frac{1}{2}}\bfG^{-1} {\bfU^{(\mcP)}}^{\rmH} \sum_{s=1}^{S} \left( p_{k,k}^{s}\bfR_{\bfh_{k,k}}^{(\mcP)}+\frac{1}{P}\bfI_{\frac{M}{K}}  \right) {\bfU^{(\mcP)}} \bfG^{-1} \bflambda_{k,k}^{-\frac{1}{2}}}\right)\\
&=DS \label{bgphimean}
\end{align}
where 
\begin{align}
\mathbb{E}\{ \tilde{\bfgamma}_{k,k}^s \tilde{\bfgamma}_{k,k}^{s \rmH}\} = \bfG^{-1} {\bfU^{(\mcP)}}^{\rmH}  \left( p_{k,k}^{s}\bfR_{\bfh_{k,k}}^{(\mcP)}+\frac{1}{P}\bfI_{\frac{M}{K}}  \right) {\bfU^{(\mcP)}} \bfG^{-1}
\end{align}
and (\ref{bgphimean}) is obtained using the fact $\sum_{s=1}^{S} p_{k,\ell}^{s}=N$.

\section{Proof of Theorem \ref{themdeq2}} \label{app2}
The mean rate loss can be decomposed into the following two terms
%\begin{align}
%&\Delta R < \nonumber\\
%& \frac{1}{NT} \sum_{t \in \mcT}\sum_{k,i}  
%{ {\rm{log}}_{2} \left(  1+ \mathbb{E}\left[ 
%		\sum_{k \neq \ell}\tilde{J} _{k,k}^{i,j}[m]+ \sum_{i \neq j}\sum_{k=1}^{d_\ell}\tilde{J} _{k,\ell}^{i,j}[m]
%		\right]  \right)}\nonumber \\
%&+\frac{1}{NT} \sum_{t \in \mcT}\sum_{k,i}  { {\rm{log}}_{2} \left(1+ \mathbb{E}\left[ 
%		\sum_{k \neq \ell}\hat{J} _{k,k}^{i,j}[m]+ \sum_{i \neq j}\sum_{k=1}^{d_\ell}\hat{J} _{k,\ell}^{i,j}[m]
%		\right]  \right)}  \label{rlossIQ}
%\end{align} 
\begin{figure*}[!t]
\begin{align}
\Delta R_{\rm ub}<&  \frac{1}{NT} \sum_{k}\sum_{m \in \mcT} {d_k}  {\rm{log}}_{2} \left(1+ NP \left(K-\frac{1}{d_k}\right) {\rm MSE}\left[m,D\right] \right)+\nonumber\\ 
&\frac{1}{NT} \sum_{k}\sum_{m \in \mcT} {d_k}  {\rm{log}}_{2} \left(1+ NPDS \left(K-\frac{1}{d_k}\right) \rho[m] Q(N_{\rmd}) \right) \label{rlossIQ}
\end{align}
\hrulefill
\vspace*{4pt}
\end{figure*}
in (\ref{rlossIQ}) due to $\log(1+A+B)<\log(1+A)+\log(1+B)$ if $A,B>0$. The first term and second term of (\ref{rlossIQ}) are caused by estimation (prediction) error and quantization error, respectively. If the number of feedback bits per channel is $N_{\rm d}= (DS-1) \log_{2} P$, the interference power due to quantization error $\hat{J} _{k,k}^{i,j}[m]$ and $\hat{J} _{k,\ell}^{i,j}[m]$ can be upper bounded by a finite value independent of $P$. Accordingly, the rate loss due to quantization error is also upper bounded.
\end{appendices}
    
\bibliographystyle{IEEEtran}
\begin{footnotesize}
\bibliography{IEEEabrv,library}

\begin{thebibliography}{10}
\providecommand{\url}[1]{#1}
\csname url@samestyle\endcsname
\providecommand{\newblock}{\relax}
\providecommand{\bibinfo}[2]{#2}
\providecommand{\BIBentrySTDinterwordspacing}{\spaceskip=0pt\relax}
\providecommand{\BIBentryALTinterwordstretchfactor}{4}
\providecommand{\BIBentryALTinterwordspacing}{\spaceskip=\fontdimen2\font plus
\BIBentryALTinterwordstretchfactor\fontdimen3\font minus
  \fontdimen4\font\relax}
\providecommand{\BIBforeignlanguage}[2]{{%
\expandafter\ifx\csname l@#1\endcsname\relax
\typeout{** WARNING: IEEEtran.bst: No hyphenation pattern has been}%
\typeout{** loaded for the language `#1'. Using the pattern for}%
\typeout{** the default language instead.}%
\else
\language=\csname l@#1\endcsname
\fi
#2}}
\providecommand{\BIBdecl}{\relax}
\BIBdecl

\bibitem{Cadambe2008a}
V.~Cadambe and S.~Jafar, ``{Interference alignment and degrees of freedom of
  the K-user interference channel},'' \emph{IEEE Trans. Inf. Theory}, vol.~54,
  no.~8, pp. 3425--3441, Aug. 2008.

\bibitem{Tresch2009}
R.~Tresch and M.~Guillaud, ``{Cellular interference alignment with imperfect
  channel knowledge},'' in \emph{2009 IEEE Int. Conf. Commun. Work.}, Jun.
  2009.

\bibitem{Xie2013}
B.~Xie, Y.~Li, H.~Minn, and A.~Nosratinia, ``{Adaptive interference alignment
  with CSI uncertainty},'' \emph{IEEE Trans. Commun.}, vol.~61, no.~2, pp.
  792--801, Feb. 2013.

\bibitem{Razavi2014}
S.~M. Razavi and T.~Ratnarajah, ``{Performance analysis of interference
  alignment under CSI mismatch},'' \emph{IEEE Trans. Veh. Technol.}, vol.~63,
  no.~9, pp. 4740--4748, Nov. 2014.

\bibitem{Aquilina15}
P.~Aquilina and T.~Ratnarajah, ``{Performance analysis of IA techniques in the
  MIMO IBC with imperfect CSI},'' \emph{IEEE Trans. Commun.}, vol.~63, no.~4,
  pp. 1259--1270, Apr. 2015.

\bibitem{Bolcskei2009}
J.~Thukral and H.~B\"olcskei, ``{Interference alignment with limited
  feedback},'' in \emph{2009 IEEE Int. Symp. Inf. Theory}, Jun. 2009, pp.
  1759--1763.

\bibitem{Krishnamachari2013}
R.~Krishnamachari and M.~Varanasi, ``{Interference alignment under limited
  feedback for MIMO interference channels},'' \emph{Signal Process. IEEE
  Trans.}, vol.~61, no.~15, pp. 3908--3917, Jul. 2013.

\bibitem{Kim2012}
J.-s. Kim, S.-h. Moon, S.-R. Lee, and I.~Lee, ``{A new channel quantization
  strategy for MIMO interference alignment with limited feedback},'' \emph{IEEE
  Trans. Wirel. Commun.}, vol.~11, no.~1, pp. 358--366, Jan. 2012.

\bibitem{Ayach2012}
O.~{El Ayach} and R.~W. Heath, ``{Grassmannian differential limited feedback
  for interference alignment},'' \emph{IEEE Trans. Signal Process.}, vol.~60,
  no.~12, pp. 6481--6494, Dec. 2012.

\bibitem{Xu2015aa}
\BIBentryALTinterwordspacing
Z.~Xu, M.~Gan, and T.~Zemen, ``{On the degrees of freedom for opportunistic
  interference alignment with 1-bit feedback: The 3 cell case}.'' [Online].
  Available: \url{http://arxiv.org/abs/1501.04312}
\BIBentrySTDinterwordspacing

\bibitem{Jindal2006}
N.~Jindal, ``{MIMO broadcast channels with finite-rate feedback},'' \emph{IEEE
  Trans. Inf. Theory}, vol.~52, no.~11, pp. 5045--5060, Nov. 2006.

\bibitem{Ayach2012a}
O.~E. Ayach and R.~W. Heath, ``{Interference alignment with analog channel
  state feedback},'' \emph{IEEE Trans. Wirel. Commun.}, vol.~11, no.~2, pp.
  626--636, Feb. 2012.

\bibitem{Caire2010}
G.~Caire, N.~Jindal, M.~Kobayashi, and N.~Ravindran, ``{Multiuser MIMO
  achievable rates with downlink training and channel state feedback},''
  \emph{IEEE Trans. Inf. Theory}, vol.~56, no.~6, pp. 2845--2866, Jun. 2010.

\bibitem{Santipach2010}
W.~Santipach and M.~L. Honig, ``{Optimization of training and feedback overhead
  for beamforming over block fading channels},'' \emph{IEEE Trans. Inf.
  Theory}, vol.~56, no.~12, pp. 6103--6115, Dec. 2010.

\bibitem{ElAyach2012}
O.~{El Ayach}, A.~Lozano, and R.~W. Heath, ``{On the overhead of interference
  alignment: training, feedback, and cooperation},'' \emph{IEEE Trans. Wirel.
  Commun.}, vol.~11, no.~11, pp. 4192--4203, Nov. 2012.

\bibitem{Mungara2014}
R.~K. Mungara, G.~George, and A.~Lozano, ``{Overhead and spectral efficiency of
  pilot-Assisted interference alignment in time-selective fading channels},''
  \emph{IEEE Trans. Wirel. Commun.}, vol.~13, no.~9, pp. 4884--4895, Sep. 2014.

\bibitem{Yu2012}
H.~Yu, Y.~Sung, H.~Kim, and Y.~H. Lee, ``{Beam Tracking for Interference
  Alignment in Slowly Fading MIMO Interference Channels: A Perturbations
  Approach Under a Linear Framework},'' \emph{IEEE Trans. Signal Process.},
  vol.~60, no.~4, pp. 1910--1926, Apr. 2012.

\bibitem{Zhao2014b}
N.~Zhao, F.~R. Yu, H.~Sun, H.~Yin, A.~Nallanathan, and G.~Wang, ``{Interference
  alignment with delayed channel state information and dynamic AR-model channel
  prediction in wireless networks},'' \emph{Wirel. Networks}, vol.~21, no.~4,
  pp. 1227--1242, May 2015.

\bibitem{Zemen2007}
T.~Zemen, C.~Mecklenbrauker, F.~Kaltenberger, and B.~Fleury, ``{Minimum-energy
  band-limited predictor with dynamic subspace selection for time-variant
  flat-fading channels},'' \emph{IEEE Trans. Signal Process.}, vol.~55, no.~9,
  pp. 4534--4548, Sep. 2007.

\bibitem{Xu2014a}
Z.~Xu and T.~Zemen, ``{Time-Variant Channel Prediction for Interference
  Alignment with Limited Feedback},'' in \emph{IEEE Int. Conf. Commun. (ICC),
  Work. Small Cell 5G Networks}, 2014.

\bibitem{edfors1996}
O.~Edfors, M.~Sandell, J.-J. van~de Beek, D.~Landstr{\"{o}}m, and
  F.~Sj{\"{o}}berg., ``{An introduction to orthogonal frequency-division
  multiplexing},'' Lule{\aa} University of Technology, Tech. Rep., 1996.

\bibitem{Zemen2012}
T.~Zemen and A.~F. Molisch, ``{Adaptive reduced-rank estimation of
  nonstationary time-variant channels using subspace selection},'' \emph{IEEE
  Trans. Veh. Technol.}, vol.~61, no.~9, pp. 4042--4056, Nov. 2012.

\bibitem{Kay1993}
S.~M. Kay, \emph{Fundamentals of Statistical Signal Processing: Estimation
  Theory}.\hskip 1em plus 0.5em minus 0.4em\relax Upper Saddle River, NJ, USA:
  Prentice-Hall, Inc., 1993.

\bibitem{Dietrich2005}
F.~Dietrich and W.~Utschick, ``{Pilot-assisted channel estimation based on
  second-order statistics},'' \emph{IEEE Trans. Signal Process.}, vol.~53,
  no.~3, pp. 1178--1193, Mar. 2005.

\bibitem{Bernad2013}
L.~Bernad{\'{o}}, T.~Zemen, F.~Tufvesson, A.~F. Molisch, and C.~F.
  Mecklenbr{\"{a}}uker, ``{Delay and Doppler spreads of nonstationary vehicular
  channels for safety-relevant scenarios},'' \emph{IEEE Trans. Veh. Technol.},
  vol.~63, no.~1, pp. 82--93, Jan. 2014.

\bibitem{Hofer15}
M.~Hofer, Z.~Xu, and T.~Zemen, ``{On the optimum number of hypotheses for
  adaptive reduced-rank subspace selection},'' in \emph{IEEE Veh. Technol.
  Conf. (VTC Fall)}, Boston, 2015.

\bibitem{Thompson1982}
D.~J. Thompson, ``{Spectrum estimation and harmonic analysis},'' \emph{Proc.
  IEEE}, vol.~70, pp. 1055--1096, 1982.

\bibitem{Slepian1978}
D.~Slepian, ``{Prolate spheroidal wave functions, Fourier analysis, and
  uncertainty---V: the discrete case},'' \emph{Bell Syst. Tech. J.}, vol.~57,
  no.~5, pp. 1371--1430, 1978.

\bibitem{Minn2000}
H.~Minn and V.~Bhargava, ``{An investigation into time-domain approach for OFDM
  channel estimation},'' \emph{IEEE Trans. Broadcast.}, vol.~46, no.~4, pp.
  240--248, Dec. 2000.

\bibitem{Wan2010}
F.~Wan, W.~P. Zhu, and M.~N.~S. Swamy, ``{Semi-blind most significant tap
  detection for sparse channel estimation of OFDM systems},'' \emph{IEEE Trans.
  Circuits Syst. I-Regular Pap.}, vol.~57, no.~3, pp. 703--713, Mar. 2010.

\bibitem{Love2003}
D.~Love, R.~Heath, and T.~Strohmer, ``{Grassmannian beamforming for
  multiple-input multiple-output wireless systems},'' \emph{IEEE Trans. Inf.
  Theory}, vol.~49, no.~10, pp. 2735--2747, Oct. 2003.

\bibitem{Mukkavilli2003}
K.~K. Mukkavilli, A.~Sabharwal, E.~Erkip, and B.~Aazhang, ``{On beamforming
  with finite rate feedback in multiple-antenna systems},'' \emph{IEEE Trans.
  Inf. Theory}, vol.~49, no.~10, pp. 2562--2579, 2003.

\bibitem{Dai2008}
W.~Dai, Y.~E. Liu, and B.~Rider, ``{Quantization bounds on Grassmann manifolds
  and applications to MIMO communications},'' \emph{IEEE Trans. Inf. Theory},
  vol.~54, no.~3, pp. 1108--1123, Mar. 2008.

\bibitem{Pitaval2011}
R.~A. Pitaval, H.~L. Maattanen, K.~Schober, O.~Tirkkonen, and R.~Wichman,
  ``{Beamforming codebooks for two transmit antenna systems based on optimum
  grassmannian packings},'' \emph{IEEE Trans. Inf. Theory}, vol.~57, no.~10,
  pp. 6591--6602, 2011.

\bibitem{Au-yeung2007}
C.~Au-yeung and D.~Love, ``{On the performance of random vector quantization
  limited feedback beamforming in a MISO system},'' \emph{IEEE Trans. Wirel.
  Commun.}, vol.~6, no.~2, pp. 458--462, Feb. 2007.

\bibitem{Giannakis2006}
P.~Xia and G.~B. Giannakis, ``{Design and analysis of transmit-beamforming
  based on limited-rate feedback},'' \emph{IEEE Trans. Signal Process.},
  vol.~54, no.~5, pp. 1853--1863, May 2006.

\bibitem{Love2006}
D.~Love and R.~Heath, ``{Limited feedback diversity techniques for correlated
  channels},'' \emph{IEEE Trans. Veh. Technol.}, vol.~55, no.~2, pp. 718--722,
  Mar. 2006.

\bibitem{Raghavan13}
V.~Raghavan and V.~V. Veeravalli, ``{Ensemble properties of RVQ-based
  limited-feedback beamforming codebooks},'' \emph{IEEE Trans. Inf. Theory},
  vol.~59, no.~12, pp. 8224--8249, Dec. 2013.

\bibitem{Heath2009a}
R.~W. Heath, T.~Wu, and A.~C.~K. Soong, ``{Progressive refinement of
  beamforming vectors for high-resolution limited feedback},'' \emph{EURASIP J.
  Adv. Signal Process.}, 2009.

\bibitem{Boccardi2007}
F.~Boccardi, H.~Huang, and A.~Alexiou, ``{Hierarchical quantization and its
  application to multiuser eigenmode transmissions for MIMO broadcast channels
  with limited feedback},'' \emph{IEEE Int. Symp. Pers. Indoor Mob. Radio
  Commun. PIMRC}, 2007.

\bibitem{Fadlallah2014}
Y.~Fadlallah, K.~Amis, A.~A{\"{i}}ssa-El-Bey, and R.~Pyndiah, ``{Interference
  alignment for a multi-user SISO interference channel},'' \emph{EURASIP J.
  Wirel. Commun. Netw.}, vol. 2014, no.~1, p.~79, 2014.

\bibitem{Kim2010}
D.~Kim and M.~Torlak, ``{Optimization of interference alignment beamforming
  cectors},'' \emph{IEEE J. Sel. Areas Commun.}, vol.~28, no.~9, pp.
  1425--1434, Dec. 2010.

\bibitem{Clarke1968}
R.~H. Clarke, ``{A statistical theory of mobile-radio reception},'' \emph{Bell
  Syst. Tech. J.}, p. 957, 1968.

\bibitem{ieeephy}
{3rd Generation Partnership Project (3GPP)}, ``{Technical Specification Group
  Radio Access Network; Evolved Universal Terrestrial Radio Access (E-UTRA);
  Physical Channels and Modulation (Release 10)}.''

\bibitem{Rezaee2012}
M.~Rezaee and M.~Guillaud, ``{Limited feedback for interference alignment in
  the K-user MIMO Interference Channel},'' \emph{2012 IEEE Inf. Theory Work.},
  pp. 667--671, Sep. 2012.

\end{thebibliography}
%\bibliography{IEEEabrv,C:/Users/zhinanxu/Dropbox/ZhinanBIB/library}
%\bibliography{IEEEabrv,E:/Dropbox/ZhinanBIB/library}
\end{footnotesize}

\begin{IEEEbiography}[{\includegraphics[width=1in,height=1.25in,clip,keepaspectratio]{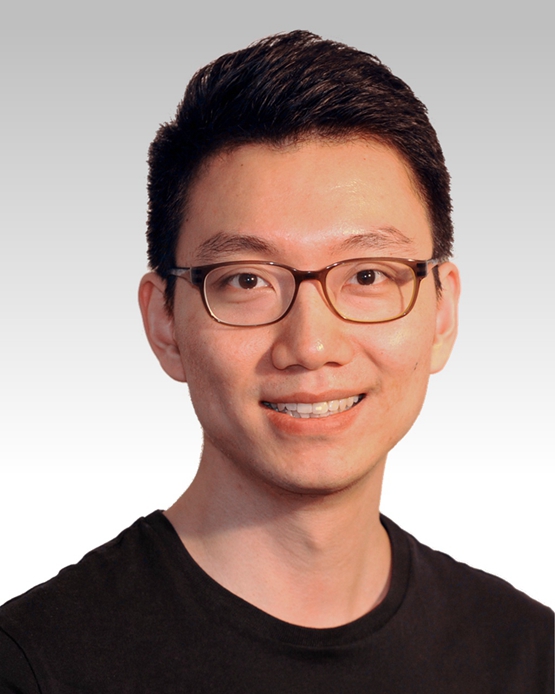}}]{Zhinan Xu}
received his M.Sc. degree in wireless communications from Lund university, Sweden in 2011 and Ph.D degree in telecommunications from Vienna
University of Technology, Austria in 2016. From 2008 to 2009, he worked as mobile network engineer with Huawei Technologies, Shenzhen, China. From 2011 to 2015, he was with the Telecommunications Research Center Vienna (FTW) working as a researcher in "Signal and Information Processing" department. Since 2015 he has been with AIT Austrian Institute of Technology as a junior scientist in the research group for ultra-reliable wireless machine-to-machine communications.
His research interests include interference management, cooperative communication systems, vehicular communications and channel modeling.
\end{IEEEbiography}
\begin{IEEEbiography}[{\includegraphics[width=1in,height=1.25in,clip,keepaspectratio]{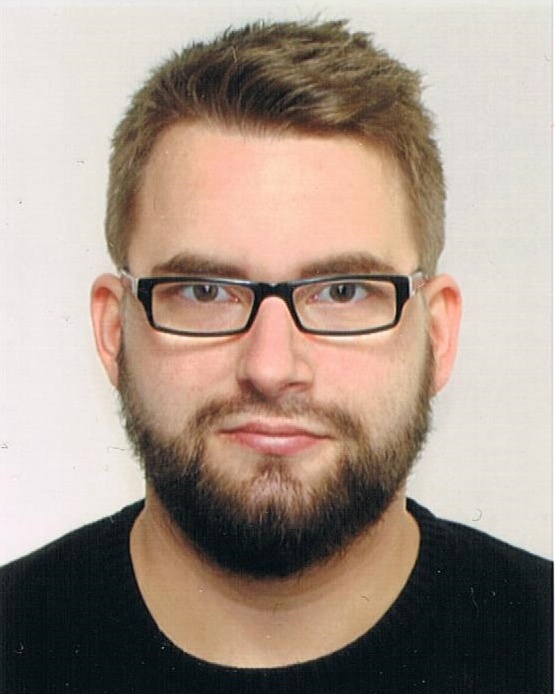}}]{Markus Hofer}
	received the Dipl.-Ing. degree (with distinction) in telecommunications from the Vienna University of Technology in 2013. Since 2013 he is working towards his Ph.D in telecommunications. From 2013 to 2015 he was with the Telecommunications Research Center Vienna (FTW) working as a researcher in "Signal and Information Processing" department. Since 2015 he is with AIT Austrian Institute of Technology as a junior scientist in the research group for ultra-reliable wireless machine-to-machine communications. 
	His research interests include low-latency wireless communications, vehicular channel measurements, modeling and emulation, time-variant channel estimation, cooperative communication systems and interference management.
\end{IEEEbiography}
\begin{IEEEbiography}[{\includegraphics[width=1in,height=1.25in,clip,keepaspectratio]{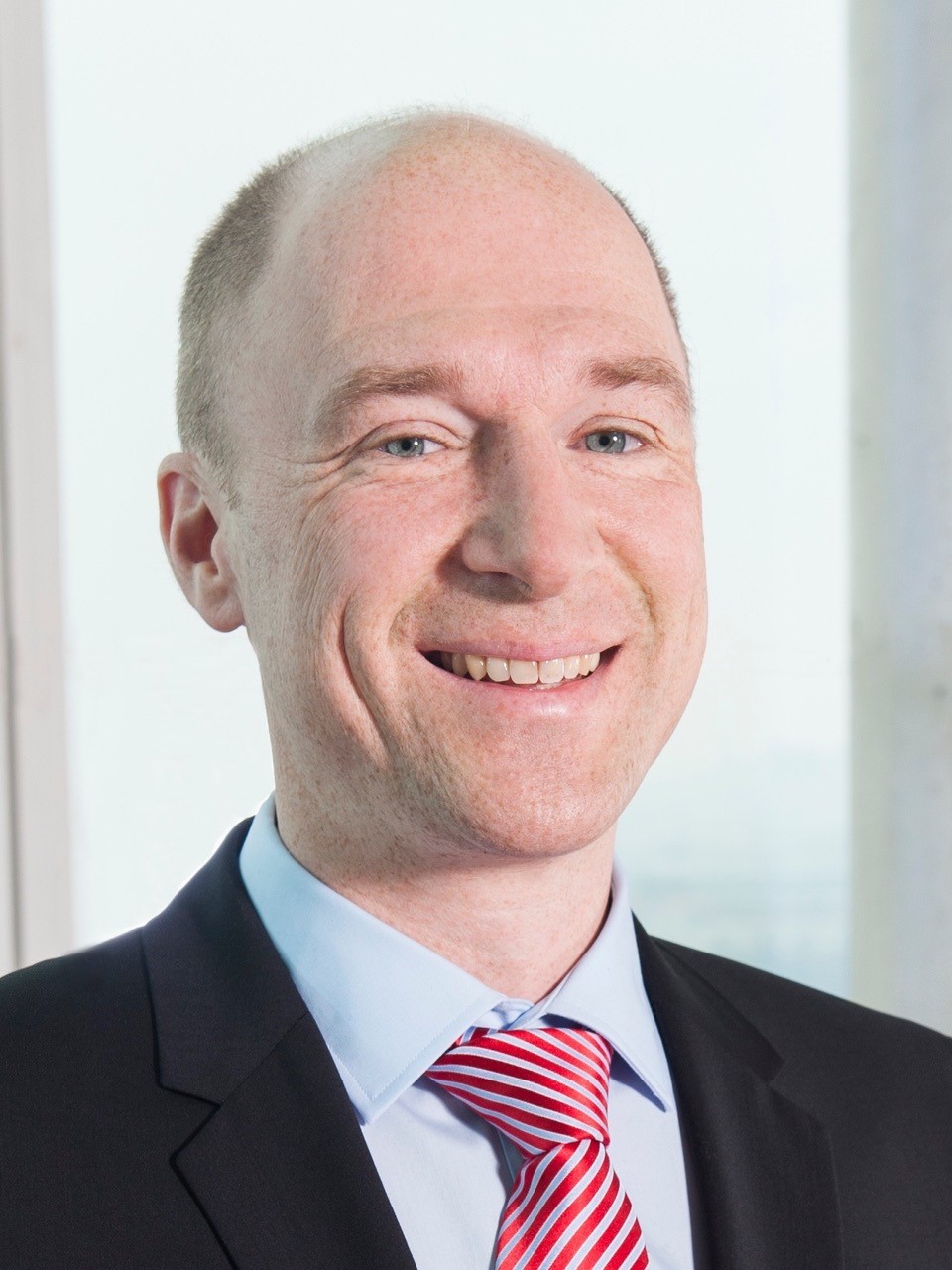}}]{Thomas Zemen}
(S'03--M'05--SM'10) received the Dipl.-Ing. degree (with distinction) in electrical engineering in 1998, the doctoral degree (with distinction) in 2004 and the Venia Docendi (Habilitation) for "Mobile Communications" in 2013, all from Vienna University of Technology.
From 1998 to 2003 he worked as Hardware Engineer and Project Manager for the Radio Communication Devices Department, Siemens Austria. From 2003 to 2015 Thomas Zemen was with FTW Telecommunications Research Center Vienna and Head of the "Signal and Information Processing" department since 2008. Since 2014 Thomas Zemen has been Senior Scientist at AIT Austrian Institute of Technology leading the research group for ultra-reliable wireless machine-to-machine communications. He is the author or coauthor of four books chapters, 30 journal papers and more than 80 conference communications. His research interests focus on reliable, low-latency wireless communications for highly autonomous vehicles; sensor and actuator networks; vehicular channel measurements and modeling; time-variant channel estimation; cooperative communication systems and interference management.
Dr. Zemen is an External Lecturer with the Vienna University of Technology and serves as Editor for the IEEE Transactions on Wireless Communications.
\end{IEEEbiography}

\end{document}